\documentclass[sigconf]{acmart}
\AtBeginDocument{%
  }

\setcopyright{acmlicensed}
\copyrightyear{2018}
\acmYear{2018}
\acmDOI{XXXXXXX.XXXXXXX}
\acmConference[Conference acronym 'XX]{Make sure to enter the correct
  conference title from your rights confirmation email}{June 03--05,
  2018}{Woodstock, NY}
\acmISBN{978-1-4503-XXXX-X/2018/06}

\usepackage{tikz}
\usepackage{amsmath}

\usepackage{filecontents}

\usepackage{lipsum}
\usepackage{tikz}
\usepackage{amsmath}
\usepackage{stfloats}
\usepackage{url}
\usepackage{verbatim}
\usepackage{graphicx}
\usepackage{dsfont}
\usepackage{amsfonts}
\usepackage{array}
\usepackage{hyperref}
\usepackage[caption=false,font=normalsize,labelfont=sf,textfont=sf]{subfig}
\usepackage{algorithm2e}
\usepackage{footmisc}
\RestyleAlgo{ruled}
\SetKwComment{Comment}{/* }{ */}
\usepackage{booktabs}
\usepackage[table]{xcolor}
\usepackage{filecontents}
\usepackage{multirow}

\usepackage{amssymb}
\usepackage{amsthm}
\usepackage{mathtools}
\usepackage{bbm}
\usepackage{bm}
\usepackage{dsfont}
\usepackage{enumitem}

\newcommand{\mcl}[1]{\mathcal{#1}}
\newcommand{\mbf}[1]{\mathbf{#1}}
\newcommand{\mbb}[1]{\mathbb{#1}}
\newcommand{\Dsyn}[0]{D_{\text{syn}}}
\newcommand{\Dreal}[0]{D_{\text{real}}}
\newcommand{\Daux}[0]{D_{\text{aux}}}

\newcommand{\Dinv}[0]{D_{\text{inv}}}
\newcommand{\adv}[0]{A}
\newcommand{\victim}[0]{V}
\newcommand{\R}[0]{\mbb R}

\newcommand{\header}[1]{\vskip 4pt \noindent \textbf{#1}}

\newcommand{\Ltraj}[0]{L_{\text{traj}}}

\newtheorem{thm}{Theorem}[section]
\newtheorem{cor}{Corollary}[section]
\newtheorem{lem}{Lemma}[section]

\numberwithin{equation}{section}

\begin{document}

\title{Turning Black Box into White Box: Dataset Distillation Leaks}


\author{Huajie Chen$^1$, Tianqing Zhu$^{\clubsuit, 1}$, Yuchen Zhong$^1$, Yang Zhang$^2$, Shang Wang$^3$, Feng He$^3$, Lefeng Zhang$^1$, Jialiang Shen$^4$, Minghao Wang$^1$, Wanlei Zhou$^1$}

\renewcommand{\shortauthors}{Chen et al.}


\begin{abstract}
Dataset distillation compresses a large real dataset into a small synthetic one, enabling models trained on the synthetic data to achieve performance comparable to those trained on the real data. Although synthetic datasets are assumed to be privacy-preserving, we show that existing distillation methods can cause severe privacy leakage because synthetic datasets implicitly encode the weight trajectories of the distilled model, they become over-informative and exploitable by adversaries.
To expose this risk, we introduce the Information Revelation Attack (IRA) against state-of-the-art distillation techniques. IRA has three stages: architecture inference, membership inference, and model inversion, each exploiting hidden information in the synthetic data.
In the architecture inference stage, the adversary trains an attack model that predicts the distillation algorithm and the model architecture using the loss trajectory—obtained by recording the losses of a model trained on the synthetic dataset. With this predicted architecture, the adversary can train a local model that mirrors the victim model’s structure and performance, \textit{effectively converting a black-box scenario into a white-box one}—surpassing the capabilities of shadow models used in prior attacks.
Next, in the membership inference stage, the adversary leverages full access to this local model to train an attack model that uses its hidden-layer and final-layer outputs to detect whether a sample was in the real dataset.
Finally, in the model inversion stage, the adversary attempts to reconstruct real samples. We propose an enhanced dual-network diffusion framework that allows imposing constraints on the generator, as well as a trajectory loss that guides the generator toward the real data distribution by exploiting deeper information embedded in the synthetic dataset.
Experiments show that IRA accurately predicts both the distillation algorithm and model architecture, and can successfully infer membership and recover sensitive samples from the real dataset.
Our code will be released after acceptance.
\end{abstract}




\maketitle

\section{Introduction}

\begin{figure}[t!]
    \centering
    \includegraphics[width=\linewidth]{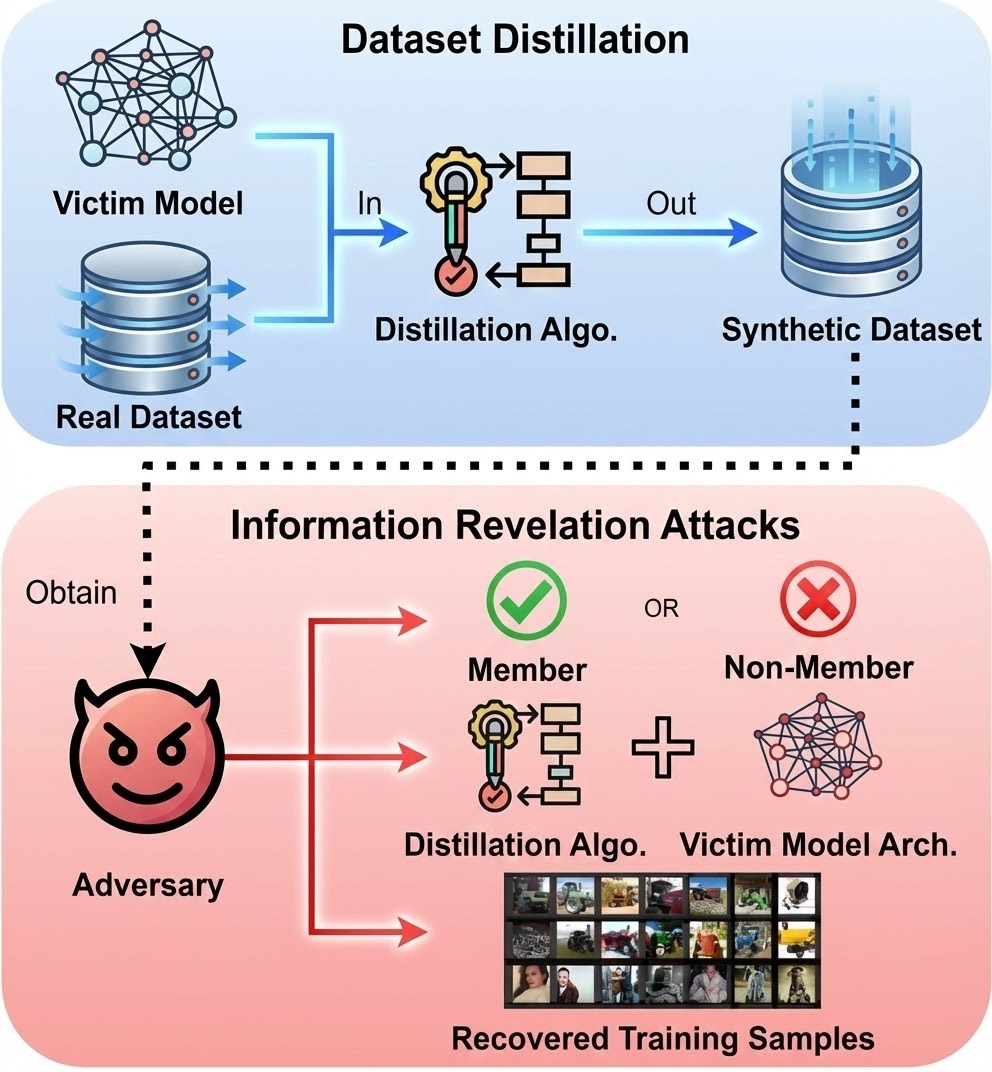}
    \caption{\textbf{The overview of IRA.} The victim generates a publicly available synthetic dataset using the distillation algorithm with the victim model and the real dataset.
    However, with the synthetic dataset, the adversary can reveal the sensitive information by launching IRA.}
    \label{fig:ddl_idea}
\end{figure}

Dataset distillation \cite{Wang2018DatasetDistillation, Zhao2021DatasetCondesation} is currently a trending technique that aims to improve the efficiency of training deep learning models using highly informative synthetic datasets.
Dataset distillation aims to distill a large real dataset down to a synthetic dataset that contains fewer samples but with a similar performance comparable to the model trained on the original dataset.
For example, Wang et al. \cite{Wang2018DatasetDistillation} synthesize a dataset that contains only $100$ synthetic images based on the CIFAR10 dataset that includes $50,000$ samples.
This dataset allows the model trained on it to perform similarly as those trained on the real dataset.
In addition, the synthetic samples are distinct from the real samples because the synthetic samples contain features for the condensed information, and therefore they tend to appear like noisy images.
Thus, it is believed that dataset distillation can be a certifiable privacy-preserving solution and will play an essential role in improving the training efficiency of deep learning models~\cite{hu2023sokPPDS,giuffre2023harnessing,lei2023comprehensive,yu2023dataset}.

Notably, state-of-the-art dataset distillation algorithms produce high-quality synthetic datasets that contain weight-trajectory information originating from the training process using the real dataset \cite{Du2023FlatTrajectoryDistillation, Guo2023TowardsLosslessDatasetDistillation, lee2024selmatch}.
Such synthetic datasets enable the trained models to achieve the same performance, or even surpass the models trained on the original dataset.
For instance, Guo et al. \cite{Guo2023TowardsLosslessDatasetDistillation} achieve lossless dataset distillation, where the CIFAR-10 dataset is distilled to a synthetic dataset including $10,000$ samples ($1,000$ images per class).
Moreover, the synthetic dataset even allows the model to reach $0.855$ test accuracy, surpassing the model trained on the real CIFAR10 dataset with $0.848$ test accuracy.

Despite the promising results and advantages offered by dataset distillation, dataset distillation is currently facing serious threats from privacy attacks.
In general settings of privacy attacks in deep learning, a victim model is usually a \textit{black box} for an adversary.
This indicates that the adversary can only query the victim model to get its outputs for further uses.
However, \textit{if the adversary can correctly infer the model architecture of the victim model, the adversary can use the high-quality synthetic dataset to train a local model that has the same architecture and similar weights as the victim model.}
\textbf{This allows the adversary to turn the black box into a white box.}
We believe that this is a severe and unexplored problem that will inevitably lead to privacy leaks and will make the synthetic dataset suffer from security and privacy attacks such as membership inference attacks \cite{shokri2017membership} and model inversion attacks \cite{fredrikson2015model}.

\header{Our Work.}
To reveal this problem, we propose \textbf{the first Information Revelation Attack (IRA) on dataset distillation methods}, as shown in \autoref{fig:ddl_idea}.
IRA consists of three stages: architecture inference, membership inference and model inversion.
In the architecture inference stage, the adversary aims to \textit{infer the distillation algorithm and the model architecture that are used in the distillation process}.
The adversary first generates synthetic datasets using different combinations of model architectures and distillation algorithms.
Then, each of these synthetic datasets is used to generate loss trajectories by training models with the same architecture.
The loss trajectories are eventually used to train an attack model that takes a loss trajectory to predict the distillation algorithm and the model architecture.
With the synthetic dataset and the predicted model architecture, the adversary can now train a local model for the next two stages.

In the membership inference stage, the goal of the adversary is \textit{to determine whether a sample belongs to the real dataset or not}.
The adversary first collects an auxiliary dataset that shares the same distribution as the real dataset.
An attack model is then trained with the outputs from the hidden and the final layers of the local model given the samples in the auxiliary dataset.
Next, given an arbitrary unseen sample, the attack model can exploit the information left in the synthetic dataset to the best of its ability to determine whether the sample is a member or a non-member.


In the model inversion stage, the adversary aims \textit{to recover data samples that can be considered as members in the real dataset}.
The adversary utilizes a dual diffusion model framework proposed by us to learn the distribution of the real dataset.
The framework enables the adversary to add constraints that force the outputs of the diffusion model to align more closely to the samples in the real dataset.
The models are trained on the auxiliary dataset along with the guidance of the local model.
The classification loss and  trajectory loss proposed in this study are designed to make the models exploit the information contained within the local model and the synthetic dataset.

To achieve the above attack goals, several research questions (RQs) need to be addressed:
\textbf{RQ1:} Why can the model architecture used in the distillation process be identified?
\textbf{RQ2:} How can the adversary collect an effective auxiliary dataset?
\textbf{RQ3:} How to determine whether a member sample is successfully recovered?

Then, we have the following answers:
\textbf{A1:} Inspired by Zhang et al. \cite{zhang2023plot}, we discover that the model architecture can be differentiated by the loss trajectories that are generated during the training process of the models with the synthetic datasets.
Moreover, we conduct a theoretical analysis, where we discuss the relation between the synthetic datasets and the corresponding loss trajectories.
We then theoretically prove the feasibility of the architecture inference attack.
\textbf{A2:} Following the manner of MIAs \cite{shokri2017membership}, the adversary is allowed to access the online datasets in public platforms such as Kaggle and Google Dataset Search.
Assuming to have the knowledge of the task, the adversary is able to collect an auxiliary dataset that is similar to or shares the same task to the original model.
\textbf{A3:} This work does not focus on discussing the metrics of evaluating the performance of any model inversion attack.
Hence, we follow the previous studies \cite{liu2024unstoppable, Kahla2022LabelOnlyMIV}, where attack accuracy and KNN distance are used as the evaluation metrics.
A recovered sample that can be classified into the target class by the victim model trained on the real dataset is then considered a member sample.
Meanwhile, such a sample with a low KNN distance is considered to be more representative in the target class.

Our contributions are listed as follows.
\begin{itemize}[itemsep=1pt, parsep=1pt]
    \item We reveal that current dataset distillation methods fail to protect neither the privacy of the real dataset nor the model used in the distillation process. 
    \item We propose the first information revelation attack on current dataset distillation methods.
    Experimental results show that IRA can effectively infer model architectures and distillation algorithms used in the distillation process. By training a local model with the correct model architecture using the synthetic dataset, the adversary can turn the black-box victim model into a white box that contributes to membership inference and model inversion attacks.
    \item We conduct a theoretical analysis of the architecture inference attack, where the relation between the synthetic dataset and the loss trajectory is thoroughly discussed. We mathematically prove the feasibility of this attack.
\end{itemize}





\section{Related Works \& Preliminaries}

Primary notations in this paper are listed in \autoref{tab:notations}.

\begin{table}[t!]
\centering
\caption{Primary Notations.}
\small
\begin{tabular}{c|l}
    \toprule
    Symbols & Definitions\\
    \midrule
    $\victim, \adv$ & The victim/adversary\\
    \cellcolor{gray!20}$f, h$ & \cellcolor{gray!20}The victim/local model\\
    $\theta, \tau$ & The weights of the victim/local model\\
    \cellcolor{gray!20}$\Dreal$, $\Dsyn$ & \cellcolor{gray!20}The real/synthetic dataset\\
    $\Daux, \Dinv$ & The auxiliary/inversed dataset\\
    \cellcolor{gray!20}$b, c$ & \cellcolor{gray!20}The membership/target class label\\
    $\mcl F$ & The set of model architectures\\
    \cellcolor{gray!20}$\Gamma$ & \cellcolor{gray!20}The set of distillation algorithms\\
    $x, y, z$ & The sample/label/random noise\\
    \cellcolor{gray!20}$L$ & \cellcolor{gray!20}The loss function\\
    $\alpha_t, \beta_t$ & The pre-defined constants\\
    \cellcolor{gray!20}$\epsilon, \lambda$ & \cellcolor{gray!20}The Gaussian noise/weight parameter\\
    $C, E$ & The number of classes/training epochs\\
    \cellcolor{gray!20}$T, l$ & \cellcolor{gray!20}The number of sampling steps/synthetic sets\\
    $M, N$ & Input dimension/Batch Size\\
    \cellcolor{gray!20}$A_M$ & \cellcolor{gray!20}The membership inference model\\
    $A_A$ & The architecture inference model\\
    \cellcolor{gray!20}$\phi, \psi$ & \cellcolor{gray!20}The dual diffusion models\\
    $T^L, T^W$ & The loss/weight trajectory\\
    \bottomrule
\end{tabular}\label{tab:notations}
\end{table}

\subsection{Dataset Distillation}
\noindent
\textbf{Dataset Distillation.}
The goal of dataset distillation \cite{Wang2018DatasetDistillation} is to create a synthetic dataset $\Dsyn$ by utilizing a large real dataset $\Dreal$, where $|\Dsyn| << |\Dreal|$.
Models trained on $\Dsyn$ will exhibit comparable performance on the test set to those trained on $\Dreal$.
The first work of dataset distillation \cite{Wang2018DatasetDistillation} aims to directly minimize the model loss on both $\Dsyn$ and $\Dreal$, which is defined as 
\begin{equation}\label{eq:dataset_distillation}
\begin{aligned}
    \Dsyn^* = & \ \underset{\Dsyn}{\arg\min} \ L(\Dreal, \theta),\\
    \text{s.t. } \theta^* = & \ \underset{\theta}{\arg\min} \ L(\Dsyn, \theta),
\end{aligned}
\end{equation}
where $\theta$ denotes the parameters of the model; $L(\cdot, \cdot)$ is the loss function.
In each iteration, $\Dsyn$ is updated by optimizing
\begin{equation}
\begin{aligned}
    \Dsyn \gets \Dsyn - \eta \nabla_{\Dsyn} L(\Dreal, \theta),
\end{aligned}
\end{equation}
where $\eta$ denotes the learning rate.

\header{Dataset Condensation.}
Dataset condensation \cite{Zhao2021DatasetCondesation} requires not only the model trained on $\Dsyn$ to have performance comparable to that of the model trained on $\Dreal$, but also the weights to have a similar convergence result compared to those of the model trained on $\Dreal$.
Let $\theta_{\Dsyn}$ and $\theta_{\Dreal}$ respectively denote the weights derived by training on $\Dsyn$ and $\Dreal$, the goal is thus defined as
\begin{equation}
\begin{aligned}
    \Dsyn^* = &\ \underset{\Dsyn}{\arg \min} \ \kappa(\theta_{\Dsyn}, \theta_{\Dreal}), \\
    \text{s.t. } \theta_{\Dsyn} = &\ \underset{\theta}{\arg \min} \ L(\Dsyn, \theta),
\end{aligned}
\end{equation}
where $\kappa(\cdot, \cdot)$ is a distance function defined as
\begin{equation}
\begin{aligned}
    \kappa(\theta_{\Dsyn}, \theta_{\Dreal}) = \sum_{i=1}^H \iota \bigg( \nabla_{\theta_{\Dsyn}^i} L(\Dsyn, \theta_{\Dsyn}),\\
    \nabla_{\theta_{\Dreal}^i} L(\Dreal, \theta_{\Dreal}) \bigg);
\end{aligned}
\end{equation}
$H$ denotes the number of layers, and $\iota(\cdot, \cdot)$ is a distance metric such as Euclidean distance.
$\Dsyn$ is updated by optimizing
\begin{equation}
\begin{aligned}
    \Dsyn \gets \Dsyn - \eta \nabla_{\Dsyn} \kappa(\theta_{\Dsyn}, \theta_{\Dreal})
\end{aligned}
\end{equation}

\header{Trajectory-Matching Dataset Distillation.}
The current dataset distillation method \cite{Guo2023TowardsLosslessDatasetDistillation, du2024SeqMatch} employs trajectory matching to perform lossless distillation.
The objective of weight-trajectory-matching dataset distillation is to make the gradient-descent trajectory of the model trained on $\Dsyn$ similar to that of the model trained on $\Dreal$, which is defined as
\begin{equation}
\begin{aligned}
    \Dsyn^* = \underset{\Dsyn}{\arg\min} \frac{\|\hat{\theta}_{t+N} - \theta^*_{t+M}\|_2^2}{\|\theta^*_t - \theta^*_{t+M}\|},
\end{aligned}
\end{equation}
where $\theta^*_t$ denotes the weights of the model trained on $\Dreal$ at step $t$; 
$M$ and $N$ are two pre-defined hyperparameters.
$\hat{\theta}_{t}$ is the weights of the model trained on $\Dsyn$ at step $t$, and is obtained by
\begin{equation}
\begin{aligned}
    \hat{\theta}_{t+i+1} = \hat{\theta}_{t+i} - \eta \nabla l(\hat{\theta}_{t+i}, \Dsyn) \text{, where } \hat{\theta}_t := \theta_t^*.
\end{aligned}
\end{equation}

\subsection{Information Revelation Attack}
\header{Architecture Inference Attacks.}
As demonstrated in previous works \cite{zhang2023plot,Gao2024Deeptheft,gao2025try}, the adversary aims to infer the model architecture of the target model $f$ in the architecture inference attack (AIA).
The model architecture and even other hyperparameters used in the training process can be inferred using side information, e.g., figures of the training loss and power side channels.
Assume there is a set of model architecture $\mcl F$ and the size of the set is denoted by $l = |\mcl F|$, the goal of the adversary is to determine which architecture is used by $f$.
Consider the side information derived in this attack as an input $x \in \mbb R^n$, the adversary attempts to find a mapping $A_A: \mbb R^n \mapsto \mbb R^l$, which is the inference model.
Similarly to MIA, training a reliable $A_A$ is crucial in AIA.

\header{Membership Inference Attacks.}
In the membership inference attack (MIA) \cite{hu2022membership}, the victim who owns a dataset $D$ trains a classifier model $f$ on $D$. 
The goal of the adversary is to determine whether a sample $\mbf x$ belongs to $D$ or not, with access to $f(\mbf x; \theta)$. 
The adversary can have either white-box access or black-box access.
Black-box access can be further divided into three subclasses: full confidence scores, top-K confidence scores, and label-only. 
A common approach in MIA \cite{shokri2017membership, Salem2019MLLeaks} is to train a local model $h$ on the same (or similar) distribution as that of the target classifier $f(\mbf x; \theta)$.
Then, the adversary trains an attack model $A_M$ based on $h$, using the outputs from it to improve $A$.
In particular, establishing an effective $A_M$ plays the most important role therein.

\header{Model Inversion Attacks.}
Model Inversion (MIV) refers to partially or entirely the reconstruction of the target model's training dataset \cite{Kahla2022LabelOnlyMIV,liu2024unstoppable}. 
This can be formulated as an optimization problem that aims to find sensitive feature values that maximize the likelihood of the target model. 
Given a training dataset $D_{\text{train}} = \{(x^1, y^1),...,(x^n, y^n)\}$ and a model $f(x; \theta)$ trained on $D_{\text{train}}$, the adversary seeks to construct an inverse dataset $D_{\text{inv}} = \{(\tilde{x}^1, \tilde{y}^1),...,(\tilde{x}^m, \tilde{y}^m)\}$, $m \leq n$, such that $\forall i \in [1, m], \exists j \in [1, n]$, $\|\tilde{x}^i - x^j\|^2 < \rho$, where $\rho \in \R$ is a pre-defined error budget.
And $\forall i \in [1, m]$, $f(\tilde{x}^i) = \tilde{y}^i$.

\subsection{Denoising Diffusion Probabilistic Models.}


Denoising diffusion probabilistic models (DDPMs) \cite{Ho2020DDPM} are a class of generative models that learn to generate data through a process of iterative noise removal, inspired by thermodynamic diffusion.
In the forward process, a clean image $x_0$ is gradually added with random Gaussian noise until the maximum step number $T$ is reached.
This forward process is described as the transition equation
\begin{equation}
\begin{aligned}
    q(x_t|x_0):=\mcl N(x_t; \sqrt{\bar{\alpha}_t}x_0, (1-\bar{\alpha}_t)\mbf I),
\end{aligned}
\end{equation}
where $\alpha_t := 1-\beta_t$;
$\bar{\alpha}_t := \prod^t_{i=1}\alpha_i$;
$\beta_t \in (0,1]$ are pre-defined constant.
The forward process is further simplified as
\begin{equation}
\begin{aligned}
    x_t = \sqrt{\bar{\alpha}_t}x_0 + \sqrt{1-\bar{\alpha}_t}\epsilon,
\end{aligned}\label{eq:xt}
\end{equation}
where $\epsilon \in \mcl N(\mbf 0, \mbf I)$.
Thereby, one can easily get $x_t$ with $x_0$ and $t \in [1, T]$.

In the reverse process, $x_T$ is stepwise de-noised into $x_0$ following the reverser transition equation
\begin{equation}\label{eq:reverse_transition}
\begin{aligned}
    p_\theta(x_{t-1}|x_t) := \mcl N(x_{t-1}; \mu_\theta(x_t, t), \Sigma_\theta(x_t, t)).
\end{aligned}
\end{equation}
The loss function used in training a DDPM is simplified to be
\begin{equation}\label{eq:ddpm_train_loss}
\begin{aligned}
    L_t := \frac{1}{2\beta_t^2} \|\tilde{\mu}_t(x_t, x_0, t) - \mu_\theta(x_t, t)\|^2,
\end{aligned}
\end{equation}
where
\begin{equation}\label{eq:mu_t}
\begin{aligned}
    \tilde{\mu}_t(x_t, x_0, t) = \frac{\sqrt{\bar{\alpha}_{t-1}}\beta_t}{1-\bar{\alpha}_t}x_0 + \frac{\sqrt{\alpha_t}(1-\bar{\alpha}_{t-1})}{1-\bar{\alpha}_t}x_t.
\end{aligned}
\end{equation}

\section{Information Revelation Attack on Dataset Distillation}
\subsection{Threat Model}

\header{The Goal of the Adversary.}
The victim $\victim$ synthesizes a synthetic dataset $\Dsyn$ using an authentic dataset $\Dreal$, a victim model $f:\R^M \mapsto \R^C$ trained on $\Dreal$, and some dataset distillation algorithm $\gamma \in \Gamma$.
Here, $M$, $C$, and $\Gamma$ respectively denote the dimension of the flattened input image, the number of classes, and the set of dataset distillation algorithms.
With access to $\Dsyn$, the adversary $\adv$ aims to reveal sensitive information from $\Dreal$. 
The sensitive information is categorized into three classes:
\textit{architecture information}, \textit{membership}, and \textit{training samples in $\Dreal$}.

To reveal the dataset distillation algorithm $\gamma$ and the model architecture of $f$ used in the distillation process, $\adv$ is required to use AIA along with all possible side information to infer the algorithm of $\gamma$, such as dataset distillation \cite{Wang2018DatasetDistillation} or dataset condensation \cite{Zhao2021DatasetCondesation}, and the architecture of $f$, e.g., AlexNet;
to reveal the membership information, $\adv$ needs to launch MIA so as to determine whether an arbitrary data sample $x \in \Dreal$;
to reveal the training samples is to recover a dataset $\Dinv$ such that $\forall x \in \Dinv, x \in \Dreal$. 
Specifically, to determine whether a given sample $x \in \Dinv$ belongs to $\Dreal$, we follow the previous studies \cite{liu2024unstoppable, Kahla2022LabelOnlyMIV}.
If a recovered sample with a target label is classified by the victim model into the target class, this recovered sample is then considered a member.

\header{The Capability of the Adversary.}
Given the fact that the synthetic dataset $\Dsyn$ is released to the public by $\victim$, $\adv$ has full access to $\Dsyn$.
$\adv$ can access publicly available datasets to collect an auxiliary dataset $\Daux$ that has a distribution similar to that of $\Dreal$.
$\adv$ also has the knowledge about all existing model architectures and training algorithms.
$\adv$ can thus use $\Dsyn$ to train any number of local models $h_i(\cdot; \tau_i)$ with $i$ as the index, and access any weight or output from an arbitrary layer. 
$\adv$ is also capable of documenting intermediates in the training process.

\header{The Challenge of the Adversary.}
$\adv$ does not have any access to $\Dreal$. 
The model architecture, the training algorithm, and the distillation method used by $\victim$ are unknown to $\adv$, despite the fact that $\adv$ can access all publicly available architectures, algorithms and distillation methods.
Firstly, $\adv$ needs to establish a reliable binary classifier $A_M$, i.e., the membership inference attack model, so as to judge whether a sample $x \in \Dreal$;
secondly, $\adv$ must generate sufficient and valid loss trajectories to support the training of the architecture inference attack model $A_M$;
Lastly, $\adv$ has to train a generative network such that given a target class, the network can generate samples that can be classified by the victim model into the target class.
The samples should also capture the perceptual features in the real dataset and look natural.

\header{Attack Scenarios.}
We envision $\victim$ as the data owner who produces the synthetic dataset $\Dsyn$ with the authentic training dataset $\Dreal$ by training a victim model $f$ with a dataset distillation algorithm $\gamma$.
$\victim$ then release $\Dsyn$ on any public online dataset platform for free use, e.g., Kaggle\footnote{https://www.kaggle.com/datasets/amarhaiqal/aspen-hysys-distillation-column-data} and Hugging Face\footnote{https://huggingface.co/datasets/lemon-mint/bge\_distillation\_train\_dataset\_v1}, or on any market for merchant purpose.
$\adv$ can either download $\Dsyn$ for free or purchase it, and then launch the aforementioned attack.

\begin{figure*}[t!]
    \centering
    \includegraphics[width=.9\linewidth]{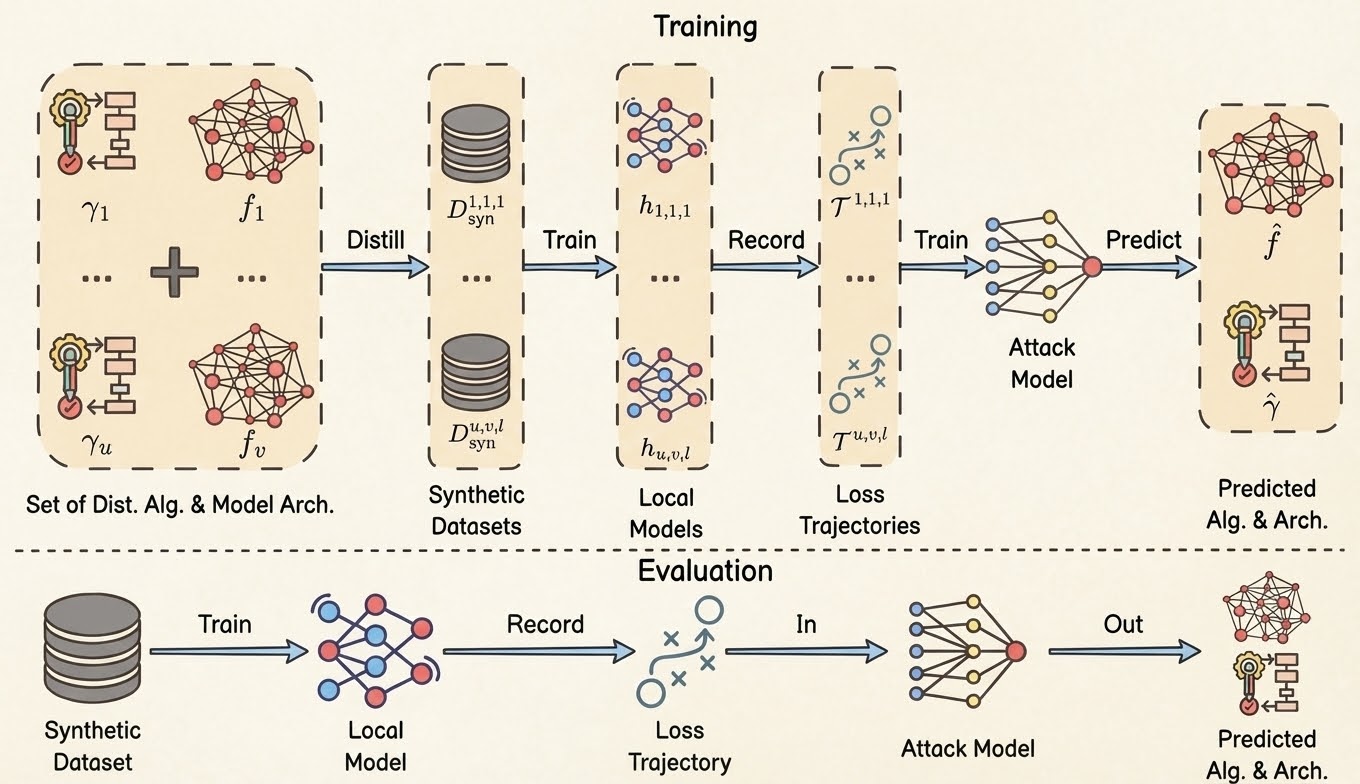}
    \caption{\textbf{Architecture Inference Stage.}
    For the AIA, $\adv$ combines different distillation algorithms $\gamma_i$ with model architectures $f_j$ to synthesize $l$ synthetic datasets.
    Then, these synthetic datasets are respectively used to train $u \times v \times l$ local models to create the set of loss trajectories $\mcl T$.
    $\mcl T$ is then used to train the architecture attack model $A_A$ to predict which algorithm and architecture are used for the dataset distillation.
    }
    \label{fig:aia_workflow}
\end{figure*}

\subsection{Initialization of Information Revelation Attack}
Having the synthetic dataset $\Dsyn$, $\adv$ begins the first step of the attack, which is to train a local model $h:\R^M \mapsto \R^C$.
There are multiple choices of hyper-parameters, model architectures, and training algorithms.
$\adv$ can even train multiple local models to form an ensemble.
Here, we consider the simplest way, which is to train one local model on a mainstream architecture, e.g., ResNet18.
By default, $\adv$ trains $h$ with the cross entropy loss 
\begin{equation}\label{eq:cross_entropy_loss}
\begin{aligned}
    L_{\text{CE}}(y, h(x)) = -\frac{1}{N} \sum_{i=1}^N \sum_{j=1}^C y_{i,j} \log(h(x)_{i,j}),
\end{aligned}
\end{equation}
where $N$ is the batch size; $y$ is the ground-truth label.
$\adv$ employs the SGD optimizer to optimize the weights.
During the training process, $\adv$ records the losses and the model weights at the end of each epoch so as to create the loss trajectory $T^L=\{L_1,...L_E\}$ of the local model.

Next, $\adv$ collects a set of model architectures $\mcl F = \{f_1,...,f_v\}$, where $v = |\mcl F|$.
Given the auxiliary dataset $\Daux$, $\adv$ uses $\gamma_i \in \Gamma$ and $f_j \in \mcl F$ to synthesize $l$ synthetic datasets denoted as $\Dsyn^{i,j,k}$, where $i$ is the index of $\gamma_i \in \Gamma$; 
$j$ denotes the index of $f_j \in \mcl F$; $k \in [1,l]$ is the index of a specific synthetic dataset.
Then, for each $\Dsyn^{i,j,k}$, $\adv$ trains a model of the same architecture on $\Dsyn^{i,j,k}$ for $E$ epochs, and records the loss trajectory $\mcl T^{i,j,k} \in \mbb R^E$.
$\adv$ eventually obtains a set $\mcl T$ that contains loss trajectories of different synthetic datasets generated by all combinations of $\gamma_i \in \Gamma$ and $f_j \in \mcl F$.


\subsection{Architecture Inference Stage}
Previous studies propose different methods in terms of matching the gradient descent trajectories or the distribution of the dataset so as to synthesize the synthetic dataset $\Dsyn$.
This results in various patterns of loss trajectories in the training of the local models. 
Therein, when a local model $h$ is trained on a $\Dsyn$ generated by a specific distillation algorithm targeting a real dataset $\Dreal$, the loss trajectory of $h$ during the training process is unique.
That is, when each $h_i$ of the same architecture is trained on each $\Dsyn^i$ distilled from the same $\Dreal$ but with different distillation algorithm, their loss trajectories $\mcl T_i$s are differentiable.
Similarly, for a fixed distillation algorithm, each $\Dsyn^i$ is synthesized by using a model $f_i$ of a specific model architecture and the same $\Dreal$.
When each of the local models $h_i$s are separately trained on each of the $\Dsyn^i$, their loss trajectories $\mcl T_i$s are also differentiable. 

Based on this finding, we design this architecture inference attack as described in \autoref{fig:aia_workflow}.
With the set of loss trajectories $\mcl T$ generated in the initialization step, for each $d_i$ and $f_j$, $\mcl A$ has $l$ loss trajectories $\mcl T^{i,j}$.
This enables $\adv$ to train an architecture inference attack model $A_A: \mbb R^E \mapsto \mbb R^{u\times v}$.
$\adv$ can use the SGD optimizer and $L_{\text{CE}}$ as illustrated in \autoref{eq:cross_entropy_loss} to train $A_A$, such that given a loss trajectory $T^L \in \mbb R^E$, $A_A(T^L) \in \mbb R^{u\times v}$ is the prediction of the distillation algorithm and the model architecture that are used in the dataset distillation process by $\victim$.
The attack process is illustrated in \autoref{alg:dd_aia}.

With the trained $A_A$, $A$ feeds the prepared loss trajectory $T^L$ of the pre-trained local model into $A_A$ to obtain the predicted distillation algorithm $\hat{\gamma}$ and model architecture $\hat{f}$.
With $\Dsyn$ and $\hat{f}$, $\adv$ can now train a local model $h$ that shares not only the same architecture with the victim model $f$, but also similar weights and performance.
This indicates that \textbf{the adversary completes turning the black box into a white box}, and is ready for the next stage.

\begin{algorithm}[t!]
\caption{Architecture Inference Stage}\label{alg:dd_aia}
\KwData{$\Dsyn$: A synthetic dataset;}
\KwResult{$\hat{i}, \hat{j}$: The indices of the distillation algo. and the model arch.;}
\tcc{Training the Attack Model}
TrajectorySet $\gets \emptyset$ \;
\tcc{For each distillation algo.}
\For{$\gamma_i \in \Gamma$}{
    \tcc{For each model arch.}
    \For{$f_j \in \mcl F$}{
        \tcc{Generate $l$ loss trajectory samples.}
        \For{$k \in [1, l]$}{
            Synthesize a $\Dsyn^{i,j,k}$
            Train a local model and record its loss trajectory $\mcl T^{ijk}$\;
            Add $(\mcl T^{ijk})$ into TrajectorySet\;
        }
    }
}
\While{not converged}{
    Train $A_A$ with TrajectorySet\;
}
\tcc{Evaluation}
Train a local model with $\Dsyn$ and record its loss trajectory $\mcl T$\;
$\hat{i}, \hat{j} \gets A_A(\mcl T)$\Comment*[r]{Make the predictions of $\hat{i}$ distillation algo. and $\hat{j}$ model arch.}
\Return $\hat{i}, \hat{j}$\;
\end{algorithm}

To further investigate the mechanism of the architecture inference attack, we conduct a theoretical analysis, where we discuss the relation between the synthetic datasets and their corresponding loss trajectories that are generated during the training process of the local model.
Weights in different neural networks of the same architecture tend to converge to a first-order stationary point, when the networks are trained on the same dataset \cite{haim2022reconstructing}.
We know that the synthetic datasets that are generated using the same algorithm and model architecture are similar to each other, whereas those generated with different algorithms or model architectures differ from each other.
Therefore, we would like to further prove that the training losses of two models of the same architecture tend to be similar, when they are separately optimized with two similar datasets.

\begin{figure*}[t!]
    \centering
    \includegraphics[width=.9\linewidth]{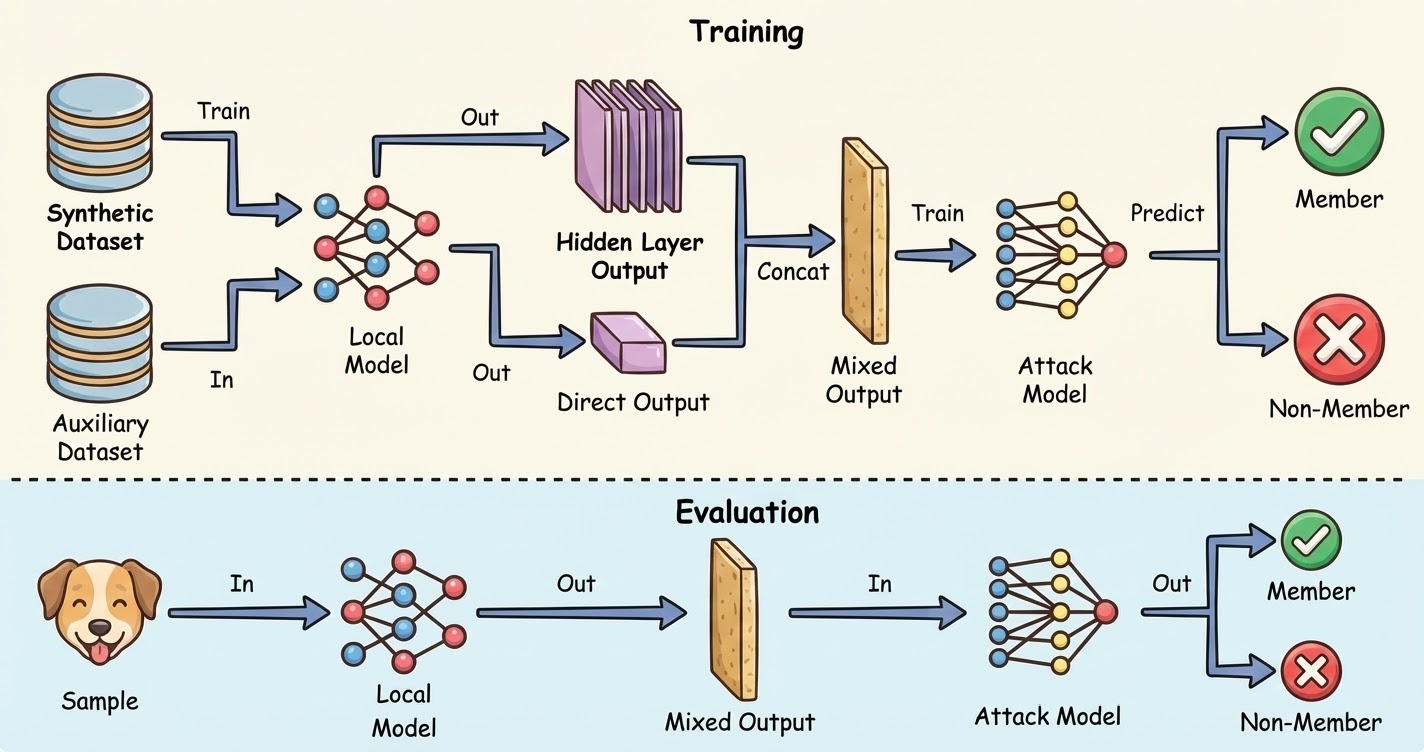}
    \caption{\textbf{Membership Inference Stage.}
    In the MIA, $\adv$ trains the membership attack model $A_M$ with the outputs from each layer of the local model $h$ to determine the membership of an arbitrary given sample $x$.
    }
    \label{fig:mia_workflow}
\end{figure*}

The following theorem shows that the distance between the weights of two neural networks of the same architecture at each epoch can be estimated by the number of epochs, the learning rate and the distance between the two datasets, when they are separately trained on two similar datasets.
Consequently, such two networks will produce losses that are close to each other at each epoch, when the distance between the datasets is sufficiently small.
\begin{thm}\label{thm:aia}
    Assume that a dataset $\mathcal{D}_2$ is a perturbed version of another dataset $\mathcal{D}_1$, i.e. 
    \begin{align}
        \mathcal{D}_2 = \{(\bm{x}_{n} + \bm{h}_{n}, y_n): \|\bm{h}_{n}\|_2 < \delta, n = 1,2,\dots,N\}, 
    \end{align}
    where $\delta > 0$ is a constant. We describe the training process on $\mathcal{D}_2$ as
    \begin{align}
        &\bm{\theta}_{2} (0) = \bm{\theta}_0 , \nonumber\\
        &\bm{\theta}_{2} (t+1) = \bm{\theta}_{2} (t)-\eta(t)\nabla \mathcal{L}_{\mathcal{D}_2}(\bm{\theta} _{2}(t)),
    \end{align}
    where $\bm{\theta}_{i}(t)$ denotes the weights in model $i$ at $t$ step; $\eta(t)$ is the learning rate at $t$ step; $L_{\mcl D_i}(\cdot)$ is the loss of the model trained on $\mcl D_i$.
        
    For any $\varepsilon > 0$, one has 
    \begin{align}
        \|\bm{\theta}_{1}(t) - \bm{\theta}_{2}(t)\| < \varepsilon ,\ t\in \{1,2,\dots,T_0\}, 
    \end{align}
    when $\delta $ is small enough. 
\end{thm}
\noindent
We can then have the following corollary:
\begin{cor}\label{cor:aia}
    For any $\varepsilon > 0$, one has  
    \begin{align}
        |\mathcal{L}_{\mathcal{D}_1}(\bm{\theta}_{1}(t)) - \mathcal{L}_{\mathcal{D}_2}(\bm{\theta}_{2}(t))| < \varepsilon ,\ t\in \{1,2,\dots,T_0\}, 
    \end{align}
    when $\delta $ is small enough. 
\end{cor}
\noindent
We prove Theorem\autoref{thm:aia} and Corollary\autoref{cor:aia} in \autoref{sect:proof} in the appendix.

According to the Corollary\autoref{cor:aia}, two similar datasets tend to result in similar loss trajectories when the local models are trained on them.
Hence, this explains why the adversary can train a classifier network with the loss trajectories to identify the dataset distillation method and the model architecture of the victim model.

\begin{algorithm}[t!]
\caption{Membership Inference Stage}\label{alg:dd_mia}
\KwData{$X$: An arbitrary data sample; $h$: A trained local model}
\KwResult{$0/1$: Non-member/Member;}
\tcc{Training the Attack Model}
TrainSet $\gets \emptyset$\;
\For{$(x, b) \in \Daux$}{Add $(h_o(x), b)$ into TrainSet\;}
\While{not converged}{
    Train $A_M$ with TrainSet\;
}
\tcc{Evaluation}
$X \sim P(D)$\Comment*[r]{Sampling}
\eIf{$A_M(X) > 0.5$}
    {\Return $1$\Comment*[r]{Member}}
    {\Return $0$\Comment*[r]{Non-member}}
\end{algorithm}

\subsection{Membership Inference Stage}


As depicted in \autoref{fig:mia_workflow}, with the well-trained local model $h(\cdot,\tau)$, $\adv$ starts to train an MIA attack model $A_M$.
In order to exploit the information left in $\Dsyn$ to the hilt, not only the output $h(x)$, the intermedia outputs of each layer in $h$ are also included in the input of $A_M$.
This contains the outputs from each layer in $h$, which is denoted as $h_o(x) = \{h^1(x), h^2(x), ..., h(x)\}$. 
$h^i(x)$ is the output of the $i$-th layer in $h$, and $h(x)$ denotes the output from the final layer in $h$.
In the training process of $A_M$, $\adv$ randomly samples from the member set and the non-member set so as to form a batch of size $N$. 
If a sample $x_i$ is in the member set, it is paired with the label $b_i=1$, otherwise $b_i=0$.
$\adv$ then trains $A_M$ with the binary cross entropy loss
\begin{equation}\label{eq:bin_ce}
\begin{aligned}
    L_{\text{BCE}}(b, A_M(h_o(x))) = & -\frac{1}{N} \sum^N_{i=1} b_i \cdot \log(A_M(h_o(x))) + \\
    & (1-b_i) \cdot \log(1 - A_M(h_o(x)))
\end{aligned}
\end{equation}
In evaluation, if $ A_M(h_o(x)) < 0.5$, we consider $x$ as a non-member sample; $x$ is determined to be a member when $A_M(h_o(x)) > 0.5$.
The attack process is illustrated in \autoref{alg:dd_mia}.

\subsection{Model Inversion Stage}


In this model inversion attack, we employ the diffusion model to be the generator network.
The main challenge in launching this attack using the diffusion model is that adding extra constraints into the loss function of the diffusion model will not result in promising outcomes.
The reason is that the DDPM outputs the predicted noise rather than the predicted clean image.
Thus, the constraints fail to pair the predicted noise with the original clean image.
To address this problem, we employ and improve the dual-network framework \cite{Benny2022DualDM} where two networks $\phi$ and $\psi$ together form the generative network as depicted in \autoref{fig:miv_workflow}.
$\phi$ is used to predict the noise $\epsilon$, and $\psi$ predicts the clean image $x_0$.
Thus, this enables $\adv$ to add constraints to the loss function of $\psi$.


During the training, given a noisy image $x_t$ generated via \autoref{eq:xt}, the pre-defined label $c$, and the uniformly sampled time step $t$, $\phi$ takes in them and predicts the noise $\epsilon_\phi(x_t, t, c)$. $c$ is set to be $\emptyset$ so as to perform unconditional sampling with a probability $p^u$. 
Thus, an estimation of $x_0$ can be derived as follows,
\begin{equation}
\begin{aligned}
    x_0^* = \frac{1}{\sqrt{\bar{\alpha}_t}} (x_t - \sqrt{1 - \bar{\alpha}_t} \epsilon_\phi(x_t, t, c)).
\end{aligned}
\end{equation}
Based on $\epsilon_\phi(x_t, t, c)$, $\adv$ obtains an estimation of the mean value $\mu_\phi(x_t, t, c)$ of $p_\phi(x_{t-1}|x_t)$ that is denoted as
\begin{equation}
\begin{aligned}
    \mu_\phi(x_t, t, c) = \frac{1}{\sqrt{\alpha_t}}x_t - \frac{1-\alpha_t}{\sqrt{1-\bar{\alpha}_t}\sqrt{\alpha_t}}\epsilon_\phi(x_t, t, c)
\end{aligned}
\end{equation}

\begin{algorithm}[t!]
\caption{Model Inversion Attack on Dataset Distillation}\label{alg:dd_miv}
\KwData{$\Dsyn$: The synthetic dataset;}
\KwResult{$x_0$: The recovered training sample;}
\tcc{Training}
\While{$L$ not converged}{
    $x_t, c \sim \Daux$\Comment*[r]{Sampling}
    $x_t \gets \sqrt{\bar{\alpha}_t}x_0 + \sqrt{1-\bar{\alpha}_t}\epsilon$\Comment*[r]{Forward}
    $\hat{\epsilon} \gets \phi(x_t, t, c)$\Comment*[r]{Noise Prediction}
    $x_0^* \gets \frac{1}{\sqrt{\bar{\alpha}_t}} (x_t - \sqrt{1 - \bar{\alpha}_t} \hat{\epsilon})$\Comment*[r]{Reverse}
    $\hat{x}_0, r_t \gets \psi(x_0^*, t, c)$\Comment*[r]{Direct Output}
    $\tilde{\mu}_t = \frac{\sqrt{\bar{\alpha}_{t-1}}\beta_t}{1-\bar{\alpha}_t}x_0 + \frac{\sqrt{\alpha_t}(1-\bar{\alpha}_{t-1})}{1-\bar{\alpha}_t}x_t$\Comment*[r]{Estimation}
    $\mu_{\phi,\psi}(x_t, t, c) = r_t \cdot \hat{x}_0 + (1-r_t)\cdot x_0^*$\Comment*[r]{Joint Output}
    $L \gets \lambda_1 \cdot L_{\text{MSE}}(\epsilon,\hat{\epsilon}) + \lambda_2 \cdot L_{\text{MSE}}(x_0, \hat{x}_0) + \lambda_3 \cdot L_{\text{MSE}}(\tilde{\mu}_t, \mu_{\phi,\psi}(x_t, t, c)) + \lambda_4 \cdot L_{\text{cls}} + \lambda_5 \cdot L_{\text{traj}}$\Comment*[r]{Compute Loss}
    Update $\phi, \psi$\;
}
\tcc{Evaluation}
$x_T \gets \mcl N(\mbf 0, \mbf I)$\Comment*[r]{Initialize Noise}
$c \gets \text{random}(0, C)$\Comment*[r]{Assign Class}
\For{$t$ in $\{T,...,1\}$}{
    $\hat{\epsilon} = \phi(x_t, t, c)$\;
    
    $\mu^S_{\hat{\epsilon}} = \frac{x_t - \sqrt{1-\bar{\alpha}_t}\hat{\epsilon}}{\sqrt{\alpha_t}} + \sqrt{1 - \bar{\alpha}_{t-1}} \cdot \hat{\epsilon}$\;
    
    $x_0^* \gets \frac{1}{\sqrt{\bar{\alpha}_t}} (x_t - \sqrt{1 - \bar{\alpha}_t} \epsilon_\phi(x_t, t, c))$\;
    $\hat{x}_0, r_t \gets \psi(x_0^*, t, c)$\;
    
    $\mu^S_{\hat{x}_0} = \sqrt{\bar{\alpha}_{t-1}}\hat{x}_0 + \sqrt{1 - \bar{\alpha}_{t-1}} \cdot \frac{x_t - \sqrt{\bar{\alpha}_t}\hat{x}_0}{\sqrt{1 - \bar{\alpha}_t}}$\;
    
    $x_{t-1} \gets r_t \cdot \mu^S_{\hat{x}_0} + (1-r_t)\cdot \mu^S_{\hat{\epsilon}}$\;
}
\Return $x_0$
\end{algorithm}

$\psi$ is then fed with $x_0^*$, $t$ and $c$ to output $\hat{x}_0$ and a dynamic value $r_t$, which is denoted as $(\hat{x}_0, r_t) = \psi(x_0^*, t, c)$.
$\hat{x}_0 \in \R^{M}$ is the predicted clean image;
$r_t \in \R^{\frac{M}{3}}$ is used to balance the impact of $\phi$ and $\psi$ on the final output $\mu_{\phi,\psi}(x_t, t)$.
The $\frac{M}{3}$ indicates that $r_t$ can be considered as a feature map having the same width and height of $\hat{x}_0$, but with only $1$ channel.
Then, $\adv$ can further obtain the estimation of $\mu_\phi(x_t, t, c)$ based on $\hat{x}_0$, denoted as 
\begin{equation}
\begin{aligned}
    \mu_{\hat{x}_0}(x_0^*, t) = \frac{\sqrt{\bar{\alpha}_{t-1}}\beta_t}{1-\bar{\alpha}_t}\hat{x}_0 + \frac{\sqrt{\alpha_t}(1-\bar{\alpha}_{t-1})}{1-\bar{\alpha}_t}x_t.
\end{aligned}
\end{equation}
Eventually, the final estimation of $\mu_{\phi,\psi}(x_t, t, c)$ is derived from
\begin{equation}
\begin{aligned}
    \mu_{\phi,\psi}(x_t, t, c) = r_t \cdot \mu_{\hat{x}_0}(x_t, t, c) + (1-r_t)\cdot \mu_{\hat{\epsilon}}(x_t, t, c)
\end{aligned}
\end{equation}

\begin{figure*}[t!]
    \centering
    \includegraphics[width=.95\textwidth]{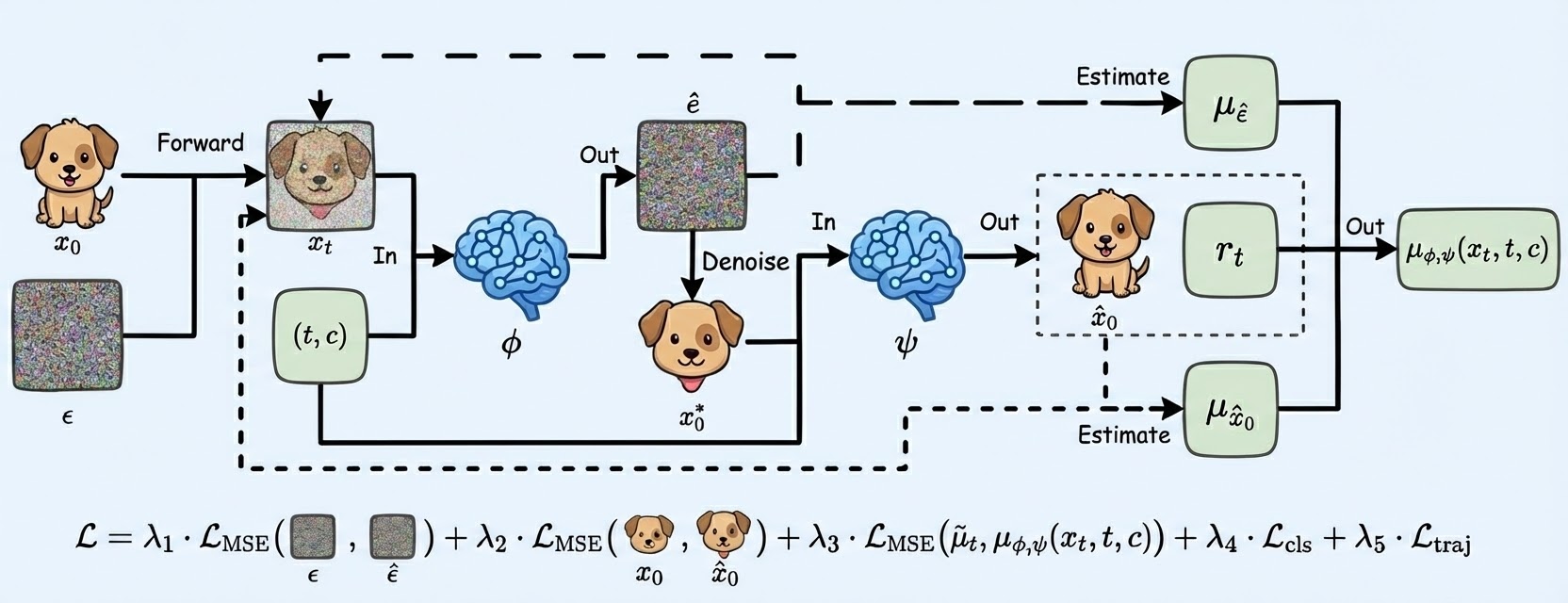}
    \caption{\textbf{Model Inversion Stage.} The framework employs two models $\phi$ and $\psi$. $\phi$ is trained to predict the noise $\epsilon$ in the DDPM way, whereas $\psi$ is trained to output the clean image $x_0$ and the coefficient $r_t$ with the denoised image $x^*_0$ as the input. The final output of this dual network framework is the combination of their outputs weighted using $r_t$. Notably, except for the first term, the rest of the terms in the loss funcion are used to regularize $\psi$.
    }
    \label{fig:miv_workflow}
\end{figure*}

To train the $\phi$ and $\psi$, $\adv$ minimizes the loss
\begin{equation}\label{eq:loss_diff_miv}
\begin{aligned}
    L = \ & \lambda_1 \cdot L_{\text{MSE}}(\epsilon,\epsilon_\phi(x_t, t, c)) + \lambda_2 \cdot L_{\text{MSE}}(x_0, \hat{x}_0)\\
    & + \lambda_3 \cdot L_{\text{MSE}}(\tilde{\mu}_t, \mu_{\phi,\psi}(x_t, t, c)) + \lambda_4 \cdot L_{\text{cls}} + \lambda_5 \cdot L_{\text{traj}}.
\end{aligned}
\end{equation}
Here, $\lambda_i$s are the weight parameters;
the first term makes $\phi$ predict $\epsilon$ more accurately;
the second term forces $\psi$ to make more precise prediction of $x_0$;
the third term is a constraint that makes $\psi$ output a more reasonable $r_t$, where $\tilde{\mu}_t = \frac{\sqrt{\bar{\alpha}_{t-1}}\beta_t}{1-\bar{\alpha}_t}x_0 + \frac{\sqrt{\alpha_t}(1-\bar{\alpha}_{t-1})}{1-\bar{\alpha}_t}x_t$;
$L_{\text{MSE}}$ denotes the MSE loss:
\begin{equation}
\begin{aligned}
    L_{\text{MSE}}(X, Y) = \frac{1}{n}\sum_i^n(x_i - y_i)^2, \forall x_i \in X, y_i \in Y;
\end{aligned}\label{eq:mse_loss}
\end{equation}
the fourth term $L_{\text{cls}}$ is the classification defined as
\begin{equation}
\begin{aligned}
    L_{\text{cls}} = L_{CE}(c, h(\hat{x}_0)).
\end{aligned}
\end{equation}
This requires the generated $\hat{x}_0$ to make $h$ yield the pre-defined label $c$, such that the diffusion model can better learn the distribution of the training dataset on which $h$ is trained;
the fifth term $L_{\text{traj}}$ is the trajectory loss, defined as
\begin{equation}
\begin{aligned}
    L_{\text{traj}}(\hat{x}_0, c, h, i) = \frac{\|\eta \nabla L_{\text{CE}}(c, h(\hat{x}_0;\tau_i^*))\|^2}{\|\tau^*_{i+1} - \tau^*_{i}\|^2}
\end{aligned}\label{eq:traj_loss}
\end{equation}
Here, $\tau^*_i$ denotes the vector of the flattened weights of $h$ derived at $i$ epoch during the training process on $\Dsyn$.
In short, by minimizing $\Ltraj$, $\adv$ aims to make the synthetic samples generated by $\mcl G$ to let $h$ produce the same loss trajectory as that of $\Dsyn$.
Notably, except for the first term, the other terms are not used to update the weights in $\phi$ but only those in $\psi$.

For sampling, $\adv$ uses the DDIM \cite{Song2021DDIM} as follows:
\begin{equation}
\begin{aligned}
    \mu_{\phi,\psi} = r_t \cdot \mu^S_{\hat{x}_0} + (1-r_t)\cdot \mu^S_{\hat{\epsilon}},
\end{aligned}
\end{equation}
where
\begin{equation}
\begin{aligned}
    \mu^S_{\hat{x}_0} = \sqrt{\bar{\alpha}_{t-1}}\hat{x}_0 + \sqrt{1 - \bar{\alpha}_{t-1}} \cdot \frac{x_t - \sqrt{\bar{\alpha}_t}\hat{x}_0}{\sqrt{1 - \bar{\alpha}_t}},
\end{aligned}
\end{equation}
and
\begin{equation}
\begin{aligned}
    \mu^S_{\hat{\epsilon}} = \frac{x_t - \sqrt{1-\bar{\alpha}_t}\hat{\epsilon}}{\sqrt{\alpha_t}} + \sqrt{1 - \bar{\alpha}_{t-1}} \cdot \hat{\epsilon}.
\end{aligned}
\end{equation}
The $\sigma_t$ is omitted here, because in DDIM, $\sigma_t = 0, \ \forall \ t$.
The attack process is illustrated in \autoref{alg:dd_miv}.

\section{Experiments}
\subsection{Settings}

\header{Datasets.}
In this experiment, we use CIFAR-10, CIFAR-100, CIFAR-100-Ext, CINIC-10, TinyImageNet-200, and ImageNet.
\begin{itemize}
    

    \item \textbf{CIFAR-10}\footnote{https://www.cs.toronto.edu/~kriz/cifar.html\label{fn:cifar}}: 
    $50,000$ samples in the training set; 
    $10,000$ samples in the test set; 
    a sample is of $3 \times 32 \times 32$ pixels;
    $10$ classes in total.

    \item \textbf{CINIC-10}\footnote{https://github.com/BayesWatch/cinic-10}: 
    $90,000$ samples in the training set; 
    $90,000$ samples in the test set; 
    a sample is of $3 \times 32 \times 32$ pixels;
    $10$ classes in total;
    it is used as an auxiliary dataset to the CIFAR-10 dataset.

    \item \textbf{CIFAR-100}\footref{fn:cifar}:
    $50,000$ samples in the training set; 
    $10,000$ samples in the test set; 
    a sample is of $3 \times 32 \times 32$ pixels;
    $100$ classes in total.

    \item \textbf{CIFAR-100-Ext}\footnote{https://www.kaggle.com/datasets/dunky11/cifar100-100x100-images-extension}: 
    $50,000$ samples in the training set; 
    $10,000$ samples in the test set; 
    a sample is of $3 \times 32 \times 32$ pixels;
    $10$ classes in total;
    it is used as an auxiliary dataset to the CIFAR-100 dataset.
    
    \item \textbf{TinyImageNet-200}\footnote{https://www.kaggle.com/c/tiny-imagenet}:
    $100,000$ samples in the training set; 
    $10,000$ samples in the validation set;
    $10,000$ samples in the test set; 
    A sample is of $3 \times 64 \times 64$ pixels;
    $200$ classes in total.

    \item \textbf{ImageNet}\footnote{https://image-net.org/}:
    $14,197,122$ annotated images in the training and validation sets; 
    Samples are cropped into $3 \times 64 \times 64$ pixels before use; 
    $1,000$ classes in total;
    it is used as an auxiliary dataset to the Tiny Imagenet dataset.
\end{itemize}


\header{Dataset Distillation Algorithms.}
The following methods are employed as the baselines for generating the synthetic datasets.
\begin{itemize}
    \item \textbf{MTT}\cite{Cazenavette2022DDMatchingTrainingTrajectories}:
    MTT optimizes the synthetic dataset such that the model trained on the synthetic dataset is guided to a similar state as that of the model trained on the real dataset. 
    This is achieved by adding a constraint that shrinks the distance between the loss trajectories of the models trained on the synthetic and the real datasets.

    \item \textbf{FTD}\cite{Du2023FlatTrajectoryDistillation}: 
    FTD optimizes the synthetic dataset such that it provides a flat trajectory for the model weights to descend.
    By inducing random noises into the teacher models, FTD performs perturbation on the weights of the teacher models on purpose so as to form the flat trajectory.
    The objective is to make the synthetic dataset robust to the weight perturbation.

    \item \textbf{DATM}\cite{Guo2023TowardsLosslessDatasetDistillation}:
    DATM emphasizes the training stage of the trajectories that are selected to match has significant effectiveness on the synthetic dataset.
    Early trajectories are more important for the low-cardinality datasets, whereas the late trajectories are more dominant for the large datasets.
    Thus, DATM aligns the difficulty of the generated patterns with the size of the synthetic dataset.
    
    \item \textbf{SelMatch}\cite{lee2024selmatch}:
    Selmatch utilizes selection-based initialization and partial updates in the trajectory matching dataset distillation process.
    This manages the difficulty level of the synthetic dataset corresponding to its size.
    
    \item \textbf{SeqMatch}\cite{du2024SeqMatch}:
    Seqmatch sequentially generates the synthetic samples so as to avoid the failure of the distilled dataset to extract the high-level features learned by the neural network in the latter epochs.
    This issue happens when a large amount of synthetic samples is being optimized.
    By adaptively optimizing the synthetic dataset, Seqmatch effectively addresses this issue.
\end{itemize}

\header{Evaluation Metrics.}
For the model trained on $\Dsyn$, we use the top-1 accuracy (Top-1 Acc.) to evaluate its performance.
Balanced accuracy (BA), true positive rate at low false positive rate (T@LF), area under curve (AUC), and Full log-scale ROC \cite{li2024seqmia, carlini2022membership,liu2021encodermi} are used to evaluate the MIA attack.
As for the AIA attack model, we use the Top-1 Acc. to measure its efficacy.
Lastly, attack accuracy (Atk. Acc.) and K-Nearest Neighbor Distance (KNN Dist.) \cite{Kahla2022LabelOnlyMIV,liu2024unstoppable,boenisch2023curious,issa2024rve} are used to evaluate the MIV attack.
Below, we illustrate the descriptions of each metric in detail.
\begin{itemize}
    \item \textbf{Top-$n$ accuracy} is the ratio of the number of the predicted classes falling within the $n$ classes having the highest confidence score to the total number of predictions.

    \item \textbf{BA} is the ratio of the number of the correct prediction made by the model to the total number of predictions on a balanced dataset consisting of member and non-members.
    
    \item \textbf{T@LF} is the TPR at a single low FPR that gives an overview of the attack performance on a small portion of samples.
    In this evaluation, T@LF is set to be the TPR at $0.1$\% FPR.
    
    \item \textbf{AUC} is the area under the receiver operating characteristic (ROC) curve that is formed by the TPR and FPR of a membership inference attack for all possible thresholds.

    \item \textbf{Full log-scale ROC} draws the ROC curve in logarithm scale so as to emphasize the true positive rates in low false positive rate regions. This provides a more comprehensive view of the attack performance.

    \item \textbf{Atk. Acc.} is defined as the ratio of the number of synthetic samples that are classified into the target category by the evaluation model to the total number of synthetic samples in the recovered dataset.
    The evaluation model is trained on the real dataset.
    
    \item \textbf{KNN Dist.} is the shortest distance between a recovered sample and the real samples in the feature space.
    This metric measures the similarity at the feature level.
    The features of the recovered and the real samples are the outputs of the penultimate layer of the evaluation model mentioned above.
    For the experiments conducted in the same dataset, we consistently use the same evaluation model that is trained to its best.
\end{itemize}

\header{Implementations.}
ConvNet, AlexNet, ResNet18, and VGG11BN are employed for victim models.
For the membership and architecture inference attack models, they share a similar 3-layer fully connected network architecture, but with different input and output dimensionality.
The UNet architecture is employed for both $\phi$ and $\psi$.
Local models are trained for $100$ epochs using the SGD optimizer with the learning rate starting at $0.01$ and decreasing by $0.2$ every $25$ epochs;
ConvNets for generating the trajectories are trained for $100$ epochs using the SGD optimizer with a learning rate $0.01$;
attack models are uniformly trained for $100$ epochs using the SGD optimizer with the learning rate starting at $0.01$.
The learning rate decreases in the cosine annealing way and finally ends at $0.0001$;
$\phi$ and $\psi$ are trained using the ADAM optimizer for $10,000$ epochs with a consistent learning rate $0.0001$.
The $\lambda$s are set to $1$ by default.
Synthetic datasets are generated with different numbers of images per class (IPC) to evaluate the impact of the number of synthetic samples on the privacy leak.

\begin{table*}[t!]
\caption{Evaluation of the IRA on the SOTA Dataset Distillation Algorithms on the CIFAR10 Dataset.}
\centering
\footnotesize
\renewcommand{\arraystretch}{1.2}
\setlength{\tabcolsep}{1.3pt}
\begin{tabular}{cc|cccc|cccc|cccc|cccc}
    \toprule
    \multirow{3}{*}{Method} & \multirow{3}{*}{Model Arch.} & \multicolumn{4}{c|}{Test Top-1 Acc. \%} & \multicolumn{4}{c|}{MIA BA$\uparrow$/AUC$\uparrow$/T@LF \%$\uparrow$} & \multicolumn{4}{c|}{AIA Top-1 Acc. \%$\uparrow$} & \multicolumn{4}{c}{MIV Atk. Acc.$\uparrow$/KNN Dist$\downarrow$}\\
    & & \multicolumn{4}{c|}{IPC} & \multicolumn{4}{c|}{IPC} & \multicolumn{4}{c|}{IPC} & \multicolumn{4}{c}{IPC}\\
    & & 100 & 250 & 500 & 1000 & 100 & 250 & 500 & 1000 & 100 & 250 & 500 & 1000 & 100 & 250 & 500 & 1000\\

    \midrule

    \multirow{4}{*}{MTT} & ConvNet & 58.4 & 60.4 & 62.8 & 63.9 & 0.77/0.85/12.5 & 0.81/0.89/21.5 & 0.81/0.89/20.1 & 0.82/0.88/19.5 & \multirow{20}{*}{79.1} & \multirow{20}{*}{80.8 } & \multirow{20}{*}{80.4} & \multirow{20}{*}{82.1} & 0.81/682.4 & 0.83/634.2 & 0.87/591.2 & 0.89/574.7\\
    & AlexNet & 56.7 & 58.2 & 60.5 & 61.3 & 0.71/0.78/4.6 & 0.72/0.80/5.9 & 0.80/0.88/20.2 & 0.80/0.88/20.2 &  &  &  &  & 0.79/653.7 & 0.82/613.2 & 0.86/588.2 & 0.86/572.8\\
    & ResNet18 & 54.6 & 59.3 & 63.1 & 64.3 & 0.70/0.77/3.2 & 0.80/0.88/17.2 & 0.80/0.87/18.2 & 0.82/0.89/19.9 &  &  &  &  & 0.77/758.7 & 0.81/603.6 & 0.87/582.3 & 0.89/552.6\\
    & VGG11-BN & 57.2 & 60.7 & 63.5 & 65.9 & 0.77/0.85/11.2 & 0.79/0.87/14.3 & 0.81/0.88/16.7 & 0.89/0.95/53.6 &  &  &  &  & 0.81/633.7 & 0.84/617.8 & 0.86/599.2 & 0.89/542.8\\

    \cline{1-10}
    \cline{15-18}
    \multirow{4}{*}{FTD} & ConvNet & 68.8 & 72.4 & 73.7 & 75.4 & 0.80/0.88/20.7 & 0.82/0.89/22.0 & 0.89/0.95/51.0 & 0.90/0.95/46.4 &  &  &  &  & 0.90/551.3 & 0.91/542.5 & 0.91/537.9 & 0.92/512.4\\
    & AlexNet & 66.3 & 68.5 & 72.8 & 73.6 & 0.82/0.89/21.6 & 0.78/0.86/13.8 & 0.89/0.95/53.5 & 0.90/0.96/55.8 &  &  &  &  & 0.89/553.4 & 0.92/508.4 & 0.92/515.6 & 0.92/509.1\\
    & ResNet18 & 67.5 & 72.1 & 74.9 & 77.8 & 0.80/0.88/16.2 & 0.82/0.89/25.0 & 0.90/0.96/57.9 & 0.89/0.95/54.1 &  &  &  &  & 0.89/546.8 & 0.90/565.8 & 0.93/502.7 & 0.93/494.5\\
    & VGG11-BN & 69.1 & 73.9 & 75.2 & 77.6 & 0.81/0.89/16.0 & 0.82/0.90/20.1 & 0.90/0.96/56.4 & 0.89/0.95/51.4 &  &  &  &  & 0.90/541.9 & 0.93/515.9 & 0.93/507.2 & 0.93/491.2\\

    \cline{1-10}
    \cline{15-18}
    \multirow{4}{*}{DATM} & ConvNet & 74.7 & 77.9 & 78.1 & 81.7 & 0.89/0.95/50.9 & 0.93/0.97/75.1 & 0.94/0.97/75.7 & 0.94/0.98/77.6 &  &  &  &  & 0.94/489.9 & 0.94/486.4 & 0.94/467.2 & 0.94/488.9\\
    & AlexNet & 71.4 & 74.6 & 77.2 & 79.2 & 0.90/0.96/51.0 & 0.89/0.95/49.5 & 0.90/0.96/53.3 & 0.93/0.97/71.0 &  &  &  &  & 0.90/569.2 & 0.93/506.2 & 0.94/468.3 & 0.94/476.1\\
    & ResNet18 & 76.4 & 78.5 & 81.6 & 82.5 & 0.93/0.97/76.0 & 0.94/0.97/74.7 & 0.94/0.98/74.9 & 0.94/0.97/72.7 &  &  &  &  & 0.94/492.1 & 0.94/487.4 & 0.94/465.8 & 0.94/472.4\\
    & VGG11-BN & 75.2 & 78.1 & 80.9 & 82.2 & 0.90/0.95/53.7 & 0.93/0.97/76.3 & 0.93/0.98/77.1 & 0.93/0.97/73.5 &  &  &  &  & 0.93/511.0 & 0.94/489.2 & 0.94/487.5 & 0.94/492.4\\

    \cline{1-10}
    \cline{15-18}
    \multirow{4}{*}{SelMatch} & ConvNet & 78.1 & 80.2 & 81.0 & 82.6 & 0.88/0.94/53.1 & 0.91/0.96/66.7 & 0.92/0.96/67.1 & 0.93/0.97/73.7 &  &  &  &  & 0.94/489.1 & 0.94/467.3 & 0.94/477.6 & 0.94/481.5\\
    & AlexNet & 75.4 & 76.1 & 78.5 & 79.8  & 0.86/0.93/37.1 & 0.91/0.96/63.3 & 0.91/0.96/64.9 & 0.93/0.97/77.6 &  &  &  &  & 0.92/521.4 & 0.94/488.9 & 0.94/471.1 & 0.94/478.3\\
    & ResNet18 & 78.2 & 82.5 & 85.3 & 83.2 & 0.88/0.95/45.7 & 0.92/0.97/68.3 & 0.92/0.97/69.1 & 0.94/0.98/74.8 &  &  &  &  & 0.94/468.2 & 0.94/477.0 & 0.94/467.9 & 0.95/457.2\\
    & VGG11-BN & 79.2 & 81.4 & 82.3 & 82.2 & 0.89/0.95/53.3 & 0.92/0.97/74.4 & 0.92/0.97/70.6 & 0.93/0.97/74.2 &  &  &  &  & 0.94/471.3 & 0.94/482.9 & 0.94/469.2 & 0.94/477.5\\

    \cline{1-10}
    \cline{15-18}
    \multirow{4}{*}{SeqMatch} & ConvNet & 68.4 & 69.2 & 72.3 & 74.2 & 0.79/0.87/17.1 & 0.81/0.89/25.6 & 0.82/0.89/18.0 & 0.85/0.92/29.6 &  &  &  &  & 0.90/552.6 & 0.90/562.7 & 0.91/547.9 & 0.91/533.2\\
    & AlexNet & 67.9 & 69.9 & 72.7 & 73.1 & 0.77/0.85/20.8 & 0.79/0.86/24.5 & 0.83/0.90/27.7 & 0.85/0.92/42.7 &  &  &  &  & 0.90/556.2 & 0.90/562.8 & 0.91/545.0 & 0.90/547.1\\
    & ResNet18 & 70.5 & 73.3 & 75.6 & 77.9 & 0.83/0.90/32.5 & 0.83/0.88/22.0 & 0.84/0.91/30.7 & 0.84/0.92/42.9 &  &  &  &  & 0.90/571.2 & 0.91/567.4 & 0.91/542.8 & 0.92/52 6.7\\
    & VGG11-BN & 71.6 & 73.4 & 76.9 & 78.4 & 0.81/0.89/20.5 & 0.82/0.89/26.4 & 0.84/0.92/40.9 & 0.86/0.93/41.1 &  &  &  &  & 0.91/572.3 & 0.90/569.4 & 0.91/554.4 & 0.91/563.1\\
    \bottomrule
\end{tabular}\label{tab:eval_cifar10}
\end{table*}

\subsection{Evaluation}
The training dataset is first randomly divided into two subsets.
The first subset is the real dataset that contains $80\%$ of the training samples. 
It is used for the dataset distillation.
The other subset consisting of the rest $20\%$ is included in the auxiliary dataset as the member samples.
Both of the subsets are class-wise balanced.
For each of the five selected dataset distillation methods, we use the four model architectures to generate synthetic datasets.
The synthetic datasets are then utilized to train the local models and generate loss trajectories.
For each class in all the $20$ combinations of the given algorithms and architectures, we generate $500$ loss trajectories.
In total, there are $10,000$ trajectories, which are divided into a training set and a test set with balanced classes according to the ratio $4:1$.
For the local models, we consistently use the ResNet18 as the model architecture.
The IRA is then launched using the local models and the loss trajectories along with the auxiliary datasets.

\header{The Architecture Inference Stage.}
As demonstrated in \autoref{tab:eval_cifar10}, \autoref{tab:eval_cifar100}, and \autoref{tab:eval_tinyimagenet}, the top-1 test accuracy of the AIA reaches over $75\%$ in most cases in the classification task, where there are $20$ classes.
This suggests that the loss trajectories generated by the same model architecture that are trained on the synthetic datasets generated by different methods are separable.
In addition, different model architectures used in the same distillation process also result in separable loss trajectories.
We also discover that the varying IPC does not have significant effect on the classification results.

Notably, the synthetic datasets generated via different methods will naturally contain the gradient trajectory information by optimizing the corresponding objective function.
For example, a flat trajectory is the objective of FTD \cite{Du2023FlatTrajectoryDistillation}, while a range of weight trajectories is selected to be the optimization objective in DATM \cite{Guo2023TowardsLosslessDatasetDistillation}.
 Thus, the loss trajectories of the synthetic datasets generated by these various methods are separable because of the intrinsic different optimization objectives of the methods.
 For the loss trajectories of the same method with different architectures, the weight trajectory information differs because of the different model architectures used in the distillation process.
 A similar phenomenon can also be found in \cite{zhang2023plot}, where different model architectures show distinctive patterns that the adversary can further exploit.

\begin{figure}[t!]
    \centering
    \includegraphics[width=\linewidth]{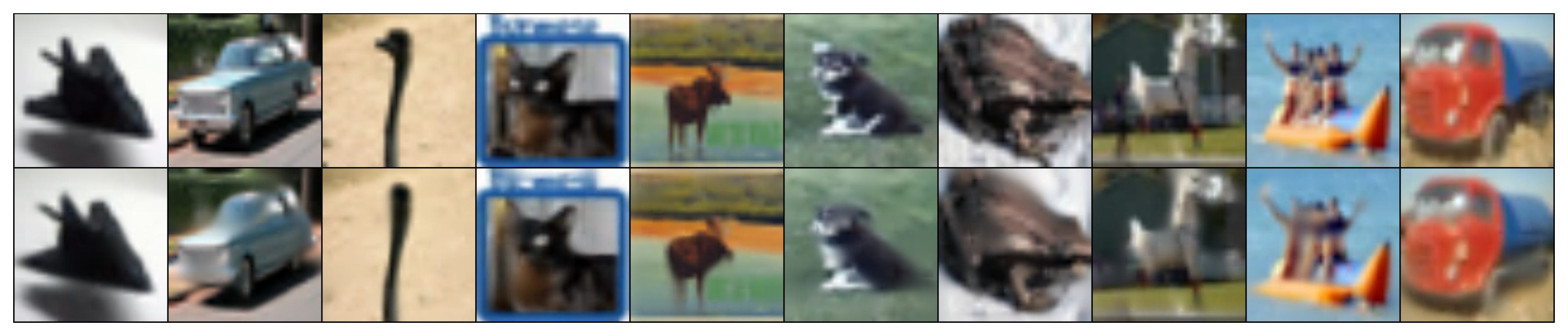}
    \caption{\textbf{MIV Qualitative Results on CIFAR-10.} The upper row are the samples in each class from the real dataset. The bottom row are the synthetic samples generated by the diffusion model.}
    \label{fig:miv_res_cifar10}
\end{figure}

\header{The Membership Inference Stage.}
The MIA attack evaluation results are listed in \autoref{tab:eval_cifar10}, \autoref{tab:eval_cifar100}, and \autoref{tab:eval_tinyimagenet}.
The full log-scale ROC curves are depicted in \autoref{fig:atk_eval_cifar10}, \autoref{fig:atk_eval_cifar100}, and \autoref{fig:atk_eval_tinyimagenet}.
First, we observe that the IPC has a significant impact on the test accuracy of the local model, whereas the selection of the model architecture in the distillation process weighs relatively less.
A higher IPC leads to a better test accuracy in each single experiment among these three datasets.

Further, we also notice that for a local model that has better performance, the BA, AUC, and T@LF of the MIA tend to be higher.
This indicates that a high-quality synthetic dataset is more likely to cause privacy leakage.
For example, the SelMatch dataset algorithm with IPC=$1,000$ and the ResNet18 model architecture can generate a high-quality synthetic dataset that enables the local model to have $83.2\%$ test accuracy on the test dataset of CIFAR-10.
However, the MIA attack performance on this model is also surprisingly high, where the BA$=0.94$, AUC=$0.98$, and T@LF$=74.8$.

Compared with the performance of the other SOTA MIAs in different scenarios \cite{li2024seqmia,das2025blind,shang2025defending}, where the adversary only has black-box access to the local model, the evaluation results in the MIA of IRA are surprisingly high.
We believe that the following facts lead to these results:
The victim synthesizes the synthetic dataset using a high-fidelity dataset distillation algorithm.
Due to the inherent loss function in the distillation process, such a distilled dataset is able to guide the locally trained local model to have a gradient descent trajectory similar to that of the victim model trained directly with the real dataset.
By releasing this synthetic dataset to the public, the victim grants white-box access to the synthetic dataset to the adversary.
The adversary is thus able to train a local model that has not only comparable performance, but also similar weights to the victim model.
In other words, \textbf{it is the victim who indirectly gives out the victim model to the adversary}.
With white-box access to the local model, the adversary can then exploit all the information within.
The output of each layer in the local model can now be used to train the attack model.
Together, these facts contribute to the high attack performance.

\begin{figure*}[t!]
    \centering
    \includegraphics[width=.95\textwidth]{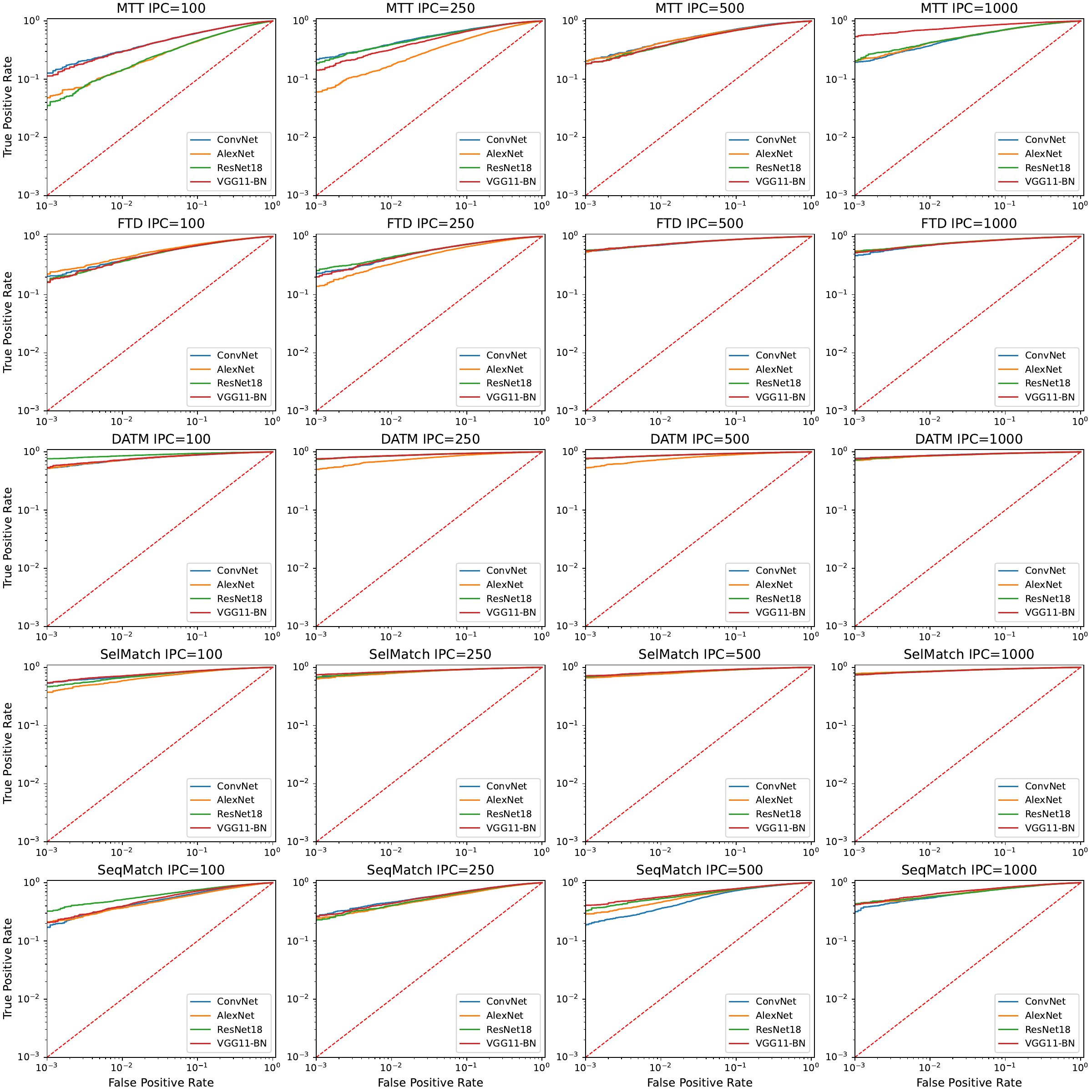}
    \caption{Log-scale ROC curves for the MIAs on the synthetic CIFAR-10 dataset.}
    \label{fig:atk_eval_cifar10}
\end{figure*}

\header{The MIV Attack.}
In \autoref{tab:eval_cifar10}, \autoref{tab:eval_cifar100}, and \autoref{tab:eval_tinyimagenet}, we list the quantitative results of the MIV attacks.
We observe that on the same real dataset, the attack accuracy tend to increase along with the ascending test accuracy of the local model, while the KNN distance decreases.
This indicates that a high-quality synthetic dataset is more likely to incur privacy leakage.
The different model architectures used in the same method have trivial impact on the MIV attack results.
This is because the MIV attack is launched on the local model with the consistent model architecture.
When the classification task becomes more complex, e.g., the number of classes increases from $10$ to $100$ in CIFAR-10 or $200$ in Tiny ImageNet, the attack performance decreases.
The degradation of the attack performance is also observed in \cite{liu2024unstoppable,Kahla2022LabelOnlyMIV}, where the increasing complexity and dimensionality of the dataset may lead to this phenomenon.

The qualitative experimental results are shown in \autoref{fig:miv_res_cifar10}, \autoref{fig:miv_res_cifar100}, and \autoref{fig:miv_res_tinyimagenet}.
These results show that the MIV model is able to generate high-quality samples by class.
The synthetic samples are realistic and can capture the subtle details and features.

\subsection{Ablation Study}
In this section, we investigate how the several key factors of the IRA attack affect the attack performance.

\header{Hidden Outputs used in the Input.}
To examine how the hidden outputs from the local model that are used in the MIA as the input affect the attack performance, we launch the MIA on the three datasets using the SelMatch dataset distillation algorithm and the ResNet18 model architecture with IPC$=100$.
The experimental results are listed in \autoref{tab:ablation_features}, where a significant drop of the attack performance can be observed when the hidden outputs are not included in the input of the attack model.
Apart from the information provided by the confidence vector, namely, the logits from the final layer, the hidden outputs give extra clues for inferring the membership.
This suggests that the features greatly contribute to the high attack performance, because the advantage of the white-box access to the local model is fully utilized.

\begin{table}[t!]
    \centering
    \renewcommand{\arraystretch}{1.}
    \setlength{\tabcolsep}{10pt}
    \caption{The impact of using the features as the input in the membership inference attack.}
    \begin{tabular}{c|ccccc}
    \toprule
      \multirow{2}{*}{Datasets} & \multicolumn{3}{c}{W/O Features as Inputs  } \\
        & BA & AUC & T@LF \% \\
    \midrule
    CIFAR-10  & 0.88/0.67 & 0.94/0.71 & 53.1/3.3  \\
    CIFAR-100  & 0.74/0.61 & 0.82/0.63 & 9.7/0.4 \\
    TinyImagenet-200 & 0.75/0.59 & 0.83/0.61 & 9.1/0.1 \\
    \bottomrule
    \end{tabular}
    \label{tab:ablation_features}
\end{table}

\header{The Loss Functions in the model inversion attack.}
Next, we investigate how the loss functions used in the model inversion attack affect the attack performance.
We consider four conditions, where the classification loss and trajectory loss are tested, respectively.
We then launch the attack on the local model that is trained on the synthetic CIFAR-10 dataset generated by the SelMatch algorithm with the ConvNet model architecture.
The experimental results can be found in \autoref{tab:ablation_loss_func}.
It is observed that without the two loss functions, the attack performance decreases greatly, because the MIV model learns only the distribution of the auxiliary dataset.
Compared to the trajectory loss, the classification loss has greater impact on the attack performance, while the trajectory loss also contributes to the enhancement of the attack.

\begin{table}[t!]
    \centering
    \renewcommand{\arraystretch}{1.}
    \setlength{\tabcolsep}{10pt}
    \caption{The impact of the loss functions used in the model inversion attack on the CIFAR-10 synthetic dataset.}
    \begin{tabular}{cc|cc}
    \toprule
      \multicolumn{2}{c|}{Loss Functions} & \multirow{2}{*}{Atk. Acc.} & \multirow{2}{*}{KNN Dist} \\
       $L_{\text{cls}}$ & $L_{\text{traj}}$ &  &  \\
    \midrule
    $\times$ & $\times$ & 0.47 & 1048.3 \\
    $\checkmark$ & $\times$ & 0.91 & 593.2 \\
    $\times$ & $\checkmark$ & 0.68 & 832.4 \\
    $\checkmark$ & $\checkmark$ & 0.94 & 489.1 \\
    \bottomrule
    \end{tabular}
    \label{tab:ablation_loss_func}
\end{table}

\section{Discussion}

Traditional dataset distillation \cite{Wang2018DatasetDistillation} or dataset condensation \cite{Zhao2021DatasetCondesation} may contribute to privacy-preserving \cite{Dong2022Privacy}, because the synthetic dataset does not contain that much sensitive information, or enable the model trained on it to share the similar performance.
However, this conclusion does not fit for the SOTA dataset distillation techniques anymore.
The quality of the synthetic dataset directly affects the attack performance of the IRA.
The direct cause of the privacy leakage may be that the SOTA dataset distillation algorithms attempt to save the weight trajectory into the synthetic dataset.
During the dataset distillation process, the victim model is trained to overfit the real dataset.
By training a local model locally using the synthetic dataset, the adversary makes the local model overfit the synthetic dataset.
Thus, the local model is likely to have similar weights to those of the victim model, which also overfits the real dataset.
This eventually leads to the success of the IRA.

As discussed in previous studies \cite{stadler2022synthetic,giomi2023unified}, synthetic datasets produced by dataset distillation are not perfect solutions to privacy issues.
If synthetic data preserves the characteristics needed for utility, it simultaneously enables privacy attacks. 
Conversely, when it provides meaningful privacy protection, it does so at significant utility costs, often by suppressing the very signals that make outlier analysis valuable.
Moreover, the privacy risk of synthetic data is not uniform across all records. 
Outliers, which are inherently more unique and easier to distinguish, are the ``Achilles' heels" of a dataset \cite{meeus2023achilles}.
Therefore, it necessitates further research in privacy-preserving synthetic dataset generation.

Potential defenses such as DP-SGD \cite{Abadi2016DPSGD,jordon2018pate} could be applied to the training process of the victim model during the dataset distillation.
By adding noises to the gradients during the training process, the weight trajectory of the victim model deviates from the original one, which is controlled by the privacy budget.
Similarly, the perturbation can also be directly applied to the synthetic samples using differential privacy \cite{zhang2024bounded}.
Furthermore, the current dataset distillation methods prefer soft labels, because soft labels help control the gradient flow.
It is thus possible to perturb the soft labels to reach the objective of privacy-preserving.
While providing protection to data privacy, this inevitably trades off the quality of the synthetic dataset with the effectiveness of privacy protection.
This suggests that it is likely impossible to generate a high-quality synthetic dataset without sacrificing the privacy.

\section{Conclusion}
We propose the information revelation attack on the current SOTA dataset distillation methods.
By fully exploiting the information within the synthetic dataset, the adversary is able to perform the membership inference, architecture inference, and model inversion attacks targeting on the real dataset and the dataset distillation algorithm.
Experimental results show that the IRA is sufficiently powerful to enable the adversary to reveal the sensitive information about the real dataset and the dataset distillation algorithm that can non-negligibly damage the interest of the owner of the dataset.
This warns the current studies related to dataset distillation that this technique per se has a great risk in privacy leakage.
Future works in this area should focus on the aspect where a privacy-preserving solution for dataset distillation should be devised.

\section*{Ethical Considerations}
In conducting this research, we strictly adhered to a "no-harm" principle. Our methodology was designed to reveal the vulnerability of current dataset distillation methods in a purely theoretical and isolated context, ensuring that no individuals, organizations, or operational systems were negatively impacted.

\header{Isolation from Operational Systems}
All experiments described in this paper were performed in a fully offline, local environment. We utilized open-source implementations of dataset distillation protocols and deployed them on our own hardware. At no point did our attack method interact with, query, or stress-test live commercial APIs or external web services. Consequently, this research caused no service disruption, latency, or financial cost to any model providers or platforms.

\header{Use of Public, Non-Sensitive Data}
The evaluation was conducted exclusively using standard, publicly available academic datasets. We did not harvest data from private users, nor did we target specific content creators. By relying solely on established benchmarks, we ensured that no personal data was processed, and no individual's privacy or intellectual property rights were infringed upon during the course of this study.

\header{Safe Security Assessment}
Our work treats the sensitive information of the synthetic datasets as a mathematical artifact to be analyzed, rather than attacking the infrastructure that hosts it. By confining our "red-teaming" efforts to a sandbox environment, we demonstrate that security vulnerabilities can be identified and documented without risk to the broader digital ecosystem or its stakeholders.

\section*{Open Science}
To promote transparency and reproducibility, we will release artifacts including all the code of all empirical studies employed in this paper. 
Our code will be released after acceptance.
All datasets we use are included in the artifact as they are public accessible.

\section*{Generative AI Usage}
In the preparation of this manuscript, we employed generative AI tools strictly for the purpose of linguistic refinement. Specifically, these tools were used to check for grammatical errors and to improve the flow and readability of the text.
We affirm that:
\begin{itemize}
    \item \textbf{No Content Generation:} The core concepts, experimental design, data analysis, and scientific conclusions are entirely the work of the authors.

    \item \textbf{No Synthetic Text:} No substantial portions of the text were generated by AI. The tools acted solely as a copy-editing assistant.

    \item \textbf{Human Oversight:} All suggestions made by the AI were manually reviewed and verified by the authors, who take full responsibility for the accuracy and integrity of the final publication.
\end{itemize}

\bibliographystyle{IEEEtran}
\bibliography{sample-base}


\clearpage
\appendix
\onecolumn

\section{Theoretical Analysis of the Architecture Inference Attack}\label{sect:proof}

We are going to discuss the relation between the dataset and the training loss. We hope that the following results will be helpful for understanding the efficacy of the proposed architecture inference attack that identifies the dataset distillation methods using the loss trajectories. See \cite{ji2018gradient,Lyu2020Gradient} for reference about theoretical analysis of training loss. 

Firstly, we list our assumptions and notations. Let $d$ be the dimension of the input vector and let $C$ be the number of classes. A dataset is a collection of several elements. An element in the dataset is a pair of an input sample vector $\bm{x}$ and the corresponding label $y\in\{1,2,\dots,C\}$. For a finite dimensional vector $\bm{a}$, $\|\bm{a}\|_2$ represents the $l_2$-norm of vector $\bm{a}$ and $B(\bm{a},\delta)$ represents the open ball with center $\bm{a}$ and radius $\delta$. 

\paragraph{Neural Network.} We consider the neural network as a vector-valued function
\begin{align*}
    \bm{f}: &\mathbb{R}^d\times \mathbb{R}^m \to \mathbb{R}^{C},
    (\bm{x},\bm{\theta})\mapsto (f_1(\bm{x},\bm{\theta}),\dots,f_{C}(\bm{x},\bm{\theta})),
\end{align*}
where $\bm{\theta}$ represents all the parameters in the neural network. 

When we train the neural network, we have fixed samples and consider $\bm{f}$ as a function of $\bm{\theta}$. Once we complete the training process, we consider $\bm{f}$ as a function of $\bm{x}$ and use it to classify the test input. However, we consider $f_i$ as a function on $\mathbb{R}^d\times \mathbb{R}^m$ in our following discussion for each $i \in \{1,2,\dots,C\}$. 

\paragraph{Loss function.} We define the loss function of $\bm{f}$ at a single sample point $(\bm{x},y)$ by the cross-entropy loss
\begin{align}
    L(\bm{x},\bm{\theta}) = -  \log \frac{e^{f_{y}(\bm{x},\bm{\theta})}}{\sum_{j = 1}^C e^{f_{j}(\bm{x},\bm{\theta})}}
\end{align}
and on a dataset $\mathcal{D}$ by the global cross-entropy loss
\begin{align}
    \mathcal{L}_{\mathcal{D}}(\bm{\theta}) := - \frac{1}{|\mathcal{D}|} \sum_{(\bm{x},y)\in Z} L(\bm{x},\bm{\theta}). 
\end{align}


\paragraph{Gradient.} The gradient of loss function $\mathcal{L}_{\mathcal{D}}$ is 
\begin{align*}
    \nabla \mathcal{L}_{\mathcal{D}}(\bm{\theta}) = \frac{1}{|\mathcal{D}|} \sum_{(\bm{x},y)\in \mathcal{D}} \left[\frac{\sum_{i = 1}^{C} \nabla f_{i}(\bm{x},\bm{\theta}) e^{f_{i}(\bm{x},\bm{\theta})}}{\sum_{i = 1}^{C} e^{f_{i}(\bm{x},\bm{\theta})}} - \nabla f_{y}(\bm{x},\bm{\theta})\right]. 
\end{align*}

\paragraph{Gradient Descent.} We assume that the loss function $\mathcal{L}_{Z}$ in the training process is a $C^1$-function and describe the gradient descent process as $\bm{\theta}(t+1) = \bm{\theta}(t)-\eta(t)\nabla \mathcal{L}_{Z}(\bm{\theta}(t))$, where $\eta(t)$ is the learning rate at time $t$.


We assume that $\mathcal{D}_1 = \{(\bm{x}_n, y_n)\}_{n=1}^{N}$ is a dataset with $N$ samples. Then, 
\begin{align*}
    \mathcal{L}_{\mathcal{D}_1}(\bm{\theta}) = - \frac{1}{N} \sum_{n = 1}^{N} \log \frac{e^{f_{y_n}(\bm{x}_n,\bm{\theta})}}{\sum_{i = 1}^{C} e^{f_{i}(\bm{x}_n,\bm{\theta})}}. 
\end{align*}
It is clear that 
\begin{align*}
    \nabla \mathcal{L}_{\mathcal{D}_1}(\bm{\theta})
    = \frac{1}{N} \sum_{n = 1}^{N} \left[\frac{\sum_{i = 1}^{C} \nabla f_{i}(\bm{x}_n,\bm{\theta}) e^{f_{i}(\bm{x}_n,\bm{\theta})}}{\sum_{i = 1}^{C} e^{f_{i}(\bm{x}_n,\bm{\theta})}} - \nabla f_{y_n}(\bm{x}_n,\bm{\theta})\right]. 
\end{align*}
We denote the initial parameters in neural network by $\bm{\theta}_0$ and describe the training process on $\mathcal{D}_1$ as 
\begin{align}\label{recursion1}
    &\bm{\theta}_1(0) = \bm{\theta}_0;\nonumber\\
    &\bm{\theta}_1(t+1) = \bm{\theta}_1(t) - \eta(t)\nabla \mathcal{L}_{\mathcal{D}_1}(\bm{\theta}_1(t)). 
\end{align}

In our following discussion, we assume the following for the neural network: 
\begin{enumerate}
    \item[(A1)] For each $i \in \{1,2,\dots,C\}$, $f_{i}$ is a $C^1$-function on $\mathbb{R}^d \times \mathbb{R}^m$; 
    \item[(A2)] For each $i \in \{1,2,\dots,C\}$, $\frac{\partial f_{i}}{\partial \theta_l}$ is locally Lipschitz on $\mathbb{R}^d \times \mathbb{R}^m$;
\end{enumerate}
We assume the following for the datasets: 
\begin{enumerate}
    \item[(A3)] There exists a constant $M_1>0$ such that $\|\bm{x}_n\|_2\leq M_1$ for all $n\in \{1,2,\dots,N\}$. 
\end{enumerate}
We assume the following for the training process on $\mathcal{D}_1$: 
\begin{enumerate} 
    \item[(A4)] There exist a suitable number of epochs $T_0\in \mathbb{N}_+$ and an appropriate constant learning rate $\eta(t) = \eta_0>0$, such that the accuracy of the neural network after $T_0$ epochs is high enough. 
    \item[(A5)] There exists a constant $M_2>0$ such that $\|\bm{\theta}_1(t)\|_2\leq M_2$ for all $t\in \{1,2,\dots,T_0\}$. 
\end{enumerate}

In most concrete missions, the collection of input vectors and the collection of parameters are bounded sets. So, it seems to be fine to set (A3) and (A5).


Let $D = \{(\bm{x}, \bm{\theta}):\|\bm{x}\|_2\leq M_1 + 1; \|\bm{\theta}\|_2\leq M_2 + 1\}$. Then $D$ is a closed and bounded subset of $\mathbb{R}^d\times \mathbb{R}^m$. Therefore, $D$ is compact on $\mathbb{R}^d\times \mathbb{R}^m$.

\begin{lem}\label{0order}
    For any $(\bm{x}', \bm{\theta}') \neq (\bm{x}'', \bm{\theta}'')\in D$, there exists $L_1> 0$, such that 
    \begin{align*}
        \left|f_i(\bm{x}', \bm{\theta}') - f_i(\bm{x}'', \bm{\theta}'')\right| \leq L_1 (\|\bm{x}'-\bm{x}''\|_2 + \|\bm{\theta}'-\bm{\theta}''\|_2),
    \end{align*}
    for each $i \in \{1,2,\dots,C\}$. 
\end{lem}

\begin{proof}
    For any fixed $i \in \{1,2,\dots,C\}$, $f_{i}$ is a $C^1$-function on $\mathbb{R}^d \times \mathbb{R}^m$ according to (A1). 

    It is easy to verify that $\frac{\partial f_{i}}{\partial x_j}$ and $\frac{\partial f_{i}}{\partial \theta_k}$ are bounded on $D$ for any $j\in\{1,2,\dots, d\}, k\in \{1,2,\dots,m\}$. Then, we have
    \begin{align*}
        |f_i(\bm{x}', \bm{\theta}') - f_i(\bm{x}'', \bm{\theta}'')| \leq  L_1 (\|\bm{x}'-\bm{x}''\|_2 + \|\bm{\theta}'-\bm{\theta}''\|_2),
    \end{align*}
    by \cite[Theorem VII.3.9]{zbMATH05044591}, where 
    \begin{align*}
        L_1 = \max_{j,k}\left\{\max_{D} \left|\frac{\partial f_{i}}{\partial x_j}(\bm{x}, \bm{\theta})\right|, \max_{D}\left|\frac{\partial f_{i}}{\partial \theta_k}(\bm{x}, \bm{\theta})\right|\right\}. 
    \end{align*}
\end{proof}

\begin{lem}\label{1order}
    For any $(\bm{x}', \bm{\theta}') \neq (\bm{x}'', \bm{\theta}'')\in D$, there exists $L_2 > 0$, such that 
    \begin{align*}
        \left|\frac{\partial f_{i}}{\partial \theta_l}(\bm{x}', \bm{\theta}') - \frac{\partial f_{i}}{\partial \theta_l}(\bm{x}'', \bm{\theta}'')\right| \leq L_2 (\|\bm{x}'-\bm{x}''\|_2 + \|\bm{\theta}'-\bm{\theta}''\|_2),
    \end{align*}
    for each $i \in \{1,2,\dots,C\}, l\in \{1,2,\dots,m\}$. 
\end{lem}

\begin{proof}
    For any fixed $i \in \{1,2,\dots,C\}, l\in \{1,2,\dots,m\}$, $\frac{\partial f_{i}}{\partial \theta_l}$ is locally Lipschitz on $\mathbb{R}^d \times \mathbb{R}^m$ according to (A2). 

    For any $(\bm{x}, \bm{\theta})\in D$, there exists $\delta _{(\bm{x}, \bm{\theta})}$ such that $\frac{\partial f_{i}}{\partial \theta_l}$ is Lipschitz on $B((\bm{x}, \bm{\theta}); \delta _{(\bm{x}, \bm{\theta})})$ with Lipschitz constant $Lip_{(\bm{x}, \bm{\theta})}$. It is clear that $\{B((\bm{x}, \bm{\theta}); \delta _{(\bm{x}, \bm{\theta})}): (\bm{x}, \bm{\theta})\in D\}$ is an open cover of $D$. Hence, there exists a finite subcover $\{B(\bm{z}_n; \delta _{\bm{z}_n}): n = 1,2,\dots, n_0\}$. Therefore, we have 
    \begin{align*}
        \left|\frac{\partial f_{i}}{\partial \theta_l}(\bm{x}', \bm{\theta}') - \frac{\partial f_{i}}{\partial \theta_l}(\bm{x}'', \bm{\theta}'')\right|
        \leq  L_2 (\|\bm{x}'-\bm{x}''\|_2 + \|\bm{\theta}'-\bm{\theta}''\|_2),
    \end{align*}
    where $L_2 = \max_{n = 1,2,\dots, n_0} Lip_{\bm{z}_n}$. 
\end{proof}

We would like to verify that the training loss on $\mathcal{D}_2$ will be very close to $\mathcal{D}_1$ if $\mathcal{D}_2$ is close enough to $\mathcal{D}_1$. To apply the mathematical tools, we consider $\mathcal{D}_2$ as a perturbation of $\mathcal{D}_1$. Before considering the training loss, we analyze the change of parameters that determine the training loss.

\begin{thm}\label{para}
    Assume that $\mathcal{D}_2$ is a perturbation of $\mathcal{D}_1$, i.e. 
    \begin{align*}
        \mathcal{D}_2 = \{(\bm{x}_{n} + \bm{h}_{n}, y_n): \|\bm{h}_{n}\|_2 < \delta, n = 1,2,\dots,N\}, 
    \end{align*}
    where $\delta > 0$ is a constant. Describe the training process on $\mathcal{D}_2$ as
    \begin{align}\label{recursion2}
        &\bm{\theta}_{2} (0) = \bm{\theta}_0 , \nonumber\\
        &\bm{\theta}_{2} (t+1) = \bm{\theta}_{2} (t)-\eta(t)\nabla \mathcal{L}_{\mathcal{D}_2}(\bm{\theta} _{2}(t)). 
    \end{align}
        
    For any $\varepsilon > 0$, we have  
    \begin{align}\label{paracontrol}
        \|\bm{\theta}_{1}(t) - \bm{\theta}_{2}(t)\| < \varepsilon ,\ t\in \{1,2,\dots,T_0\}, 
    \end{align}
    when $\delta $ is small enough. 
\end{thm}

Then, we immediately deduce the following result about the training loss by the local continuity of the loss function at each single sample point. 

\begin{cor}\label{loss}
    For any $\varepsilon > 0$, we have 
    \begin{align}\label{losscontrol}
        |\mathcal{L}_{\mathcal{D}_1}(\bm{\theta}_{1}(t)) - \mathcal{L}_{\mathcal{D}_2}(\bm{\theta}_{2}(t))| < \varepsilon ,t\in \{0,1,\dots,T_0\}, 
    \end{align}
    when $\delta $ is small enough. 
\end{cor}

\begin{proof}[Proof of Theorem \autoref{para}]
    By \eqref{recursion1} and \eqref{recursion2}, we have the following estimate 
    \begin{align}\label{recursion3}
        \|&\bm{\theta}_{1}(t+1) - \bm{\theta}_{2}(t+1)\|_{2}\leq \|\bm{\theta}_{1} (t) - \bm{\theta}_{2} (t)\|_{2} \nonumber\\
        &+ \eta(t)\|\nabla \mathcal{L}_{\mathcal{D}_1}(\bm{\theta} _{1}(t)) - \nabla \mathcal{L}_{\mathcal{D}_2}(\bm{\theta} _{2}(t))\|_{2}, 
    \end{align}
    where $t\in \{0,1,\dots,T_0 - 1\}$. 

    According to Lemma \autoref{0order}, we have 
    \begin{align}\label{ycz2}
        \sum_{i = 1}^{C} &\left| \frac{\partial f_{i}}{\partial \theta_l} (\bm{x}_n,\bm{\theta} _{1}(t)) \right| \left|\frac{e^{f_{i}(\bm{x}_n,\bm{\theta} _{1}(t))}}{\sum_{j = 1}^{C} e^{f_{j}(\bm{x}_n,\bm{\theta} _{1}(t))}} - \frac{e^{f_{i}(\bm{x}_n + \bm{h}_n,\bm{\theta} _{2}(t))}}{\sum_{j = 1}^{C} e^{f_{j}(\bm{x}_n + \bm{h}_n,\bm{\theta} _{2}(t))}} \right|\nonumber\\
        \leq &\sum_{i = 1}^{C} \left| \frac{\partial f_{i}}{\partial \theta_l} (\bm{x}_n,\bm{\theta} _{1}(t)) \right| \frac{e^{f_{i}(\bm{x}_n,\bm{\theta} _{1}(t))}}{\sum_{j = 1}^{C} e^{f_{j}(\bm{x}_n,\bm{\theta} _{1}(t))}}\sum_{k = 1}^{C} \frac{\left|e^{f_{k}(\bm{x}_n,\bm{\theta} _{1}(t)) - f_{k}(\bm{x}_n + \bm{h}_n,\bm{\theta} _{2}(t))} - 1\right|e^{f_{k}(\bm{x}_n + \bm{h}_n,\bm{\theta} _{2}(t))}}{\sum_{j = 1}^{C} e^{f_{j}(\bm{x}_n + \bm{h}_n,\bm{\theta} _{2}(t))}}\nonumber\\
        &+ \sum_{i = 1}^{C} \left| \frac{\partial f_{i}}{\partial \theta_l} (\bm{x}_n,\bm{\theta} _{1}(t)) \right| \frac{\left| e^{f_{i}(\bm{x}_n,\bm{\theta} _{1}(t)) - f_{i}(\bm{x}_n + \bm{h}_n,\bm{\theta} _{2}(t))} - 1\right|e^{f_{i}(\bm{x}_n + \bm{h}_n,\bm{\theta} _{2}(t))}}{\sum_{j = 1}^{C} e^{f_{j}(\bm{x}_n + \bm{h}_n,\bm{\theta} _{2}(t))}}
        \leq  4 L_1^2 (\|\bm{h}_n\|_{2} + \|\bm{\theta}_{1} (t) - \bm{\theta}_{2} (t)\|_{2} ),
    \end{align}
    for each $l = \{1,2,\dots,m\}$ and $n \in \{1,2,\dots,N\}$. 
    
    According to Lemma \autoref{1order}, we have 
    \begin{align}\label{ycz4}
        \left|\frac{\partial f_{y_n}}{\partial \theta_l} (\bm{x}_n,\bm{\theta} _{1}(t)) - \frac{\partial f_{y_n}}{\partial \theta_l} (\bm{x}_n + \bm{h}_n,\bm{\theta} _{2}(t))\right| 
        \leq L_2 (\|\bm{h}_n\|_{2} + \|\bm{\theta}_{1} (t) - \bm{\theta}_{2} (t)\|_{2} ),
    \end{align}
    and
    \begin{align}\label{ycz1}
        \sum_{i = 1}^{C} \frac{\left|\frac{\partial f_{i}}{\partial \theta_l} (\bm{x}_n,\bm{\theta} _{1}(t)) - \frac{\partial f_{i}}{\partial \theta_l} (\bm{x}_n + \bm{h}_n,\bm{\theta} _{2}(t)) \right|e^{f_{i}(\bm{x}_n + \bm{h}_n,\bm{\theta} _{2}(t))}}{\sum_{j = 1}^{C} e^{f_{j}(\bm{x}_n + \bm{h}_n,\bm{\theta} _{2}(t))}}
        \leq L_2 (\|\bm{h}_n\|_{2} + \|\bm{\theta}_{1} (t) - \bm{\theta}_{2} (t)\|_{2} ),
    \end{align}
    for each $l = \{1,2,\dots,m\}$ and $n \in \{1,2,\dots,N\}$. 
    
    It follows from \eqref{ycz2} and \eqref{ycz1} that
    \begin{align}\label{ycz3}
        \left|\sum_{i = 1}^{C} \frac{ \frac{\partial f_{i}}{\partial \theta_l} (\bm{x}_n,\bm{\theta} _{1}(t)) e^{f_{i}(\bm{x}_n,\bm{\theta} _{1}(t))}}{\sum_{j = 1}^{C} e^{f_{j}(\bm{x}_n,\bm{\theta} _{1}(t))}} - \sum_{i = 1}^{C} \frac{\frac{\partial f_{i}}{\partial \theta_l} (\bm{x}_n + \bm{h}_n,\bm{\theta} _{2}(t)) e^{f_{i}(\bm{x}_n + \bm{h}_n,\bm{\theta} _{2}(t))}}{\sum_{j = 1}^{C} e^{f_{j}(\bm{x}_n + \bm{h}_n,\bm{\theta} _{2}(t))}}\right|\nonumber\\
        \leq \sum_{i = 1}^{C} \left|\frac{\partial f_{i}}{\partial \theta_l} (\bm{x}_n,\bm{\theta} _{1}(t)) - \frac{\partial f_{i}}{\partial \theta_l} (\bm{x}_n + \bm{h}_n,\bm{\theta} _{2}(t)) \right|\frac{e^{f_{i}(\bm{x}_n + \bm{h}_n,\bm{\theta} _{2}(t))}}{\sum_{j = 1}^{C} e^{f_{j}(\bm{x}_n + \bm{h}_n,\bm{\theta} _{2}(t))}}\nonumber\\
        + \sum_{i = 1}^{C} \left| \frac{\partial f_{i}}{\partial \theta_l} (\bm{x}_n,\bm{\theta} _{1}(t)) \right| \left|\frac{e^{f_{i}(\bm{x}_n,\bm{\theta} _{1}(t))}}{\sum_{j = 1}^{C} e^{f_{j}(\bm{x}_n,\bm{\theta} _{1}(t))}} - \frac{e^{f_{i}(\bm{x}_n + \bm{h}_n,\bm{\theta} _{2}(t))}}{\sum_{j = 1}^{C} e^{f_{j}(\bm{x}_n + \bm{h}_n,\bm{\theta} _{2}(t))}} \right|\nonumber\\
        \leq L_2 (\|\bm{h}_n\|_{2} + \|\bm{\theta}_{1} (t) - \bm{\theta}_{2} (t)\|_{2} ) + 4 L_1^2 (\|\bm{h}_n\|_{2} + \|\bm{\theta}_{1} (t) - \bm{\theta}_{2} (t)\|_{2} ), 
    \end{align}
    for each $l = \{1,2,\dots,m\}$ and $n \in \{1,2,\dots,N\}$.

    It follows from \eqref{ycz4} and \eqref{ycz3}
    \begin{align}
        \left| \frac{\partial \mathcal{L}_{\mathcal{D}_1}}{\partial \theta_l} (\bm{\theta} _{1}(t)) - \frac{\partial \mathcal{L}_{\mathcal{D}_2}}{\partial \theta_l} (\bm{\theta} _{2}(t))\right|
        \leq \frac{1}{N} \sum_{n = 1}^{N} \left|\frac{\partial f_{y_n}}{\partial \theta_l} (\bm{x}_n,\bm{\theta} _{1}(t)) - \frac{\partial f_{y_n}}{\partial \theta_l} (\bm{x}_n + \bm{h}_n,\bm{\theta} _{2}(t))\right|\nonumber\\
        + \frac{1}{N} \sum_{n = 1}^{N} \left|\sum_{i = 1}^{C} \frac{ \frac{\partial f_{i}}{\partial \theta_l} (\bm{x}_n,\bm{\theta} _{1}(t)) e^{f_{i}(\bm{x}_n,\bm{\theta} _{1}(t))}}{\sum_{j = 1}^{C} e^{f_{j}(\bm{x}_n,\bm{\theta} _{1}(t))}} - \sum_{i = 1}^{C} \frac{\frac{\partial f_{i}}{\partial \theta_l} (\bm{x}_n + \bm{h}_n,\bm{\theta} _{2}(t)) e^{f_{i}(\bm{x}_n + \bm{h}_n,\bm{\theta} _{2}(t))}}{\sum_{j = 1}^{C} e^{f_{j}(\bm{x}_n + \bm{h}_n,\bm{\theta} _{2}(t))}}\right|\nonumber\\
        \leq (4 L_1^2 + 2L_2) (\delta + \|\bm{\theta}_{1} (t) - \bm{\theta}_{2} (t)\|_{2} ).  
    \end{align}
    Furthermore, we have
    \begin{align}\label{gradestimate}
        &\|\nabla \mathcal{L}_{\mathcal{D}_1}(\bm{\theta} _{1}(t)) - \nabla \mathcal{L}_{\mathcal{D}_2}(\bm{\theta} _{2}(t))\|_{2} \leq \sqrt{m}(4 L_1^2 + 2L_2) (\delta + \|\bm{\theta}_{1} (t) - \bm{\theta}_{2} (t)\|_{2} ). 
    \end{align}
    
    Therefore, we conclude that 
    \begin{align*}
        \|\bm{\theta}_{1}(t+1) - \bm{\theta}_{2}(t+1)\|_{2} + \delta \leq [1 + \eta_0\sqrt{m}(4 L_1^2 + 2L_2)] ( \|\bm{\theta}_{1} (t) - \bm{\theta}_{2} (t)\|_{2} + \delta ),
    \end{align*}
    by \eqref{recursion3} and \eqref{gradestimate}. 
    Then, we have the following estimates
    \begin{align}\label{paraestimate}
        \|\bm{\theta}_{1} (t) - \bm{\theta}_{2} (t)\|_{2}  
        \leq [(1 + \eta_0\sqrt{m}(4 L_1^2 + 2L_2))^t-1]\delta,
        t\in \{0,1,\dots,T_0\}. 
    \end{align}
    Hence, we have \eqref{paracontrol} when
    $
        \delta < \frac{\varepsilon}{(1 + \eta_0\sqrt{m}(4 L_1^2 + 2L_2))^{T_0}-1}. 
    $
\end{proof}

\begin{proof}[Proof of Corollary \autoref{loss}]
    For any fixed epoch $t\in \{0,1,\dots,T_0\}$ and any fixed sample $(\bm{x}_n,y_n) \in \mathcal{D}_1$, the loss function at a single sample point $L$ is continuous at $(\bm{x}_n,\bm{\theta}_1(t))$. Therefore, there exists $\iota _{n,t_0}(\varepsilon)>0$ such that 
    \begin{align*}
        |L(\bm{x}',\bm{\theta}') - L(\bm{x}_n,\bm{\theta}_1(t))| < \varepsilon, 
    \end{align*}
    when $(\bm{x}',\bm{\theta}')\in B((\bm{x}_n,\bm{\theta}_1(t)),\iota_{n,t}(\varepsilon))$. 
    
    Using the estimates \eqref{paraestimate} we proved in Theorem \autoref{para}, we deduce that when 
    \begin{align}\label{condition}
        \delta < \min \left\{\frac{\iota _{n,t}(\varepsilon)}{[1 + \eta_0\sqrt{m}(4 L_1^2 + 2L_2)]^{t}}:
        n = 1,2,\dots,N, t = 0,1,\dots,T_0\right\},
    \end{align}
    we have
    \begin{align*}
        \|(\bm{x}_n+\bm{h}_n,\bm{\theta}_2(t)) - (\bm{x}_n,\bm{\theta}_1(t))\| < \iota _{n,t}(\varepsilon), 
        t\in \{0,1,\dots,T_0\},\ n \in \{1,2,\dots,N\}, 
    \end{align*}
    and consequently we have
    \begin{align*}
        |L(\bm{x}_n+\bm{h}_n,\bm{\theta}_2(t)) - L(\bm{x}_n,\bm{\theta}_1(t))| < \varepsilon, \nonumber\\
        t\in \{0,1,\dots,T_0\},\ n \in \{1,2,\dots,N\}. 
    \end{align*}
    Therefore, \eqref{losscontrol} holds when \eqref{condition} holds. 
\end{proof}


\begin{figure*}[h!]
\centering
\includegraphics[width=\linewidth]{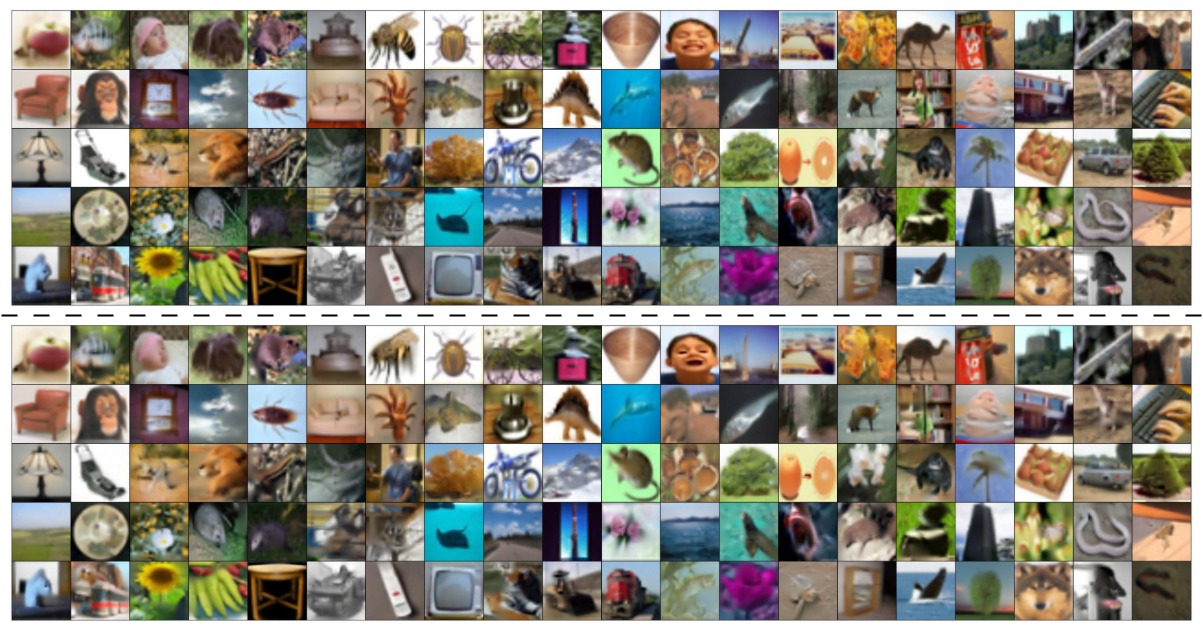}
\caption{\textbf{MIV Qualitative Results on CIFAR-100.} Real samples are in the upper part, and their corresponding recovered samples are in the lower part.}
\label{fig:miv_res_cifar100}
\end{figure*}

\begin{figure*}[h!]
\centering
\includegraphics[width=\linewidth]{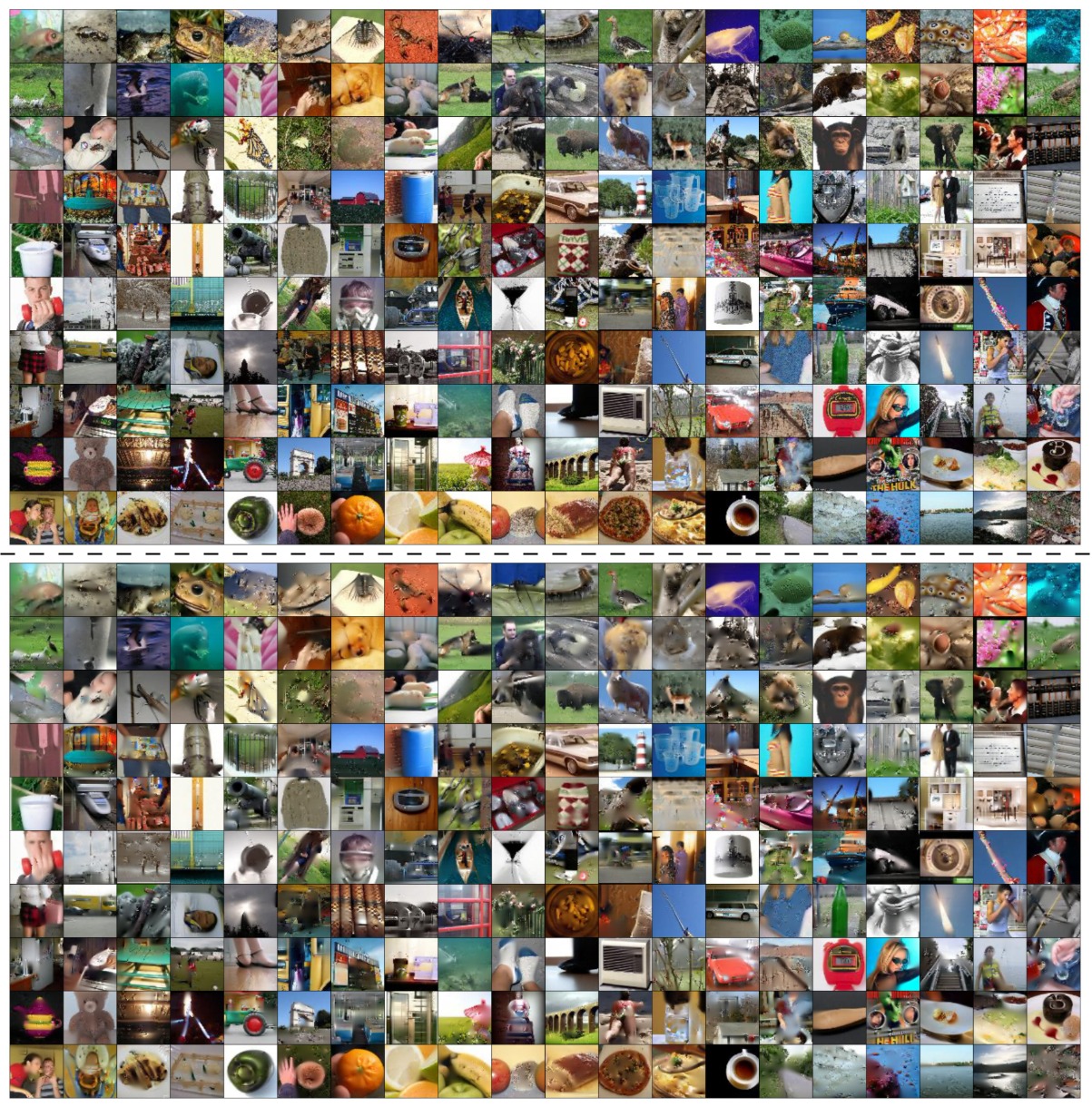}
\caption{The qualitative results of the MIV attack on TinyImagenet-200.}
\label{fig:miv_res_tinyimagenet}
\end{figure*}

\begin{table*}[h!]
\caption{Evaluation of the IRA on the SOTA Dataset Distillation Algorithms on the CIFAR-100 Dataset.}
\centering
\footnotesize
\renewcommand{\arraystretch}{1.2}
\setlength{\tabcolsep}{1.3pt}
\begin{tabular}{cc|cccc|cccc|cccc|cccc}
    \toprule
    \multirow{3}{*}{Method} & \multirow{3}{*}{Model Arch.} & \multicolumn{4}{c|}{Test Top-1 Acc \%} & \multicolumn{4}{c|}{MIA BA$\uparrow$/AUC$\uparrow$/T@LF \%$\uparrow$} & \multicolumn{4}{c|}{AIA Top-1 Acc. \%$\uparrow$} & \multicolumn{4}{c}{MIV Atk. Acc.$\uparrow$/KNN Dist$\downarrow$}\\
    & & \multicolumn{4}{c|}{IPC} & \multicolumn{4}{c|}{IPC} & \multicolumn{4}{c|}{IPC} & \multicolumn{4}{c}{IPC}\\
    & & 10 & 25 & 50 & 100 & 10 & 25 & 50 & 100 & 10 & 25 & 50 & 100 & 10 & 25 & 50 & 100\\

    \midrule

    \multirow{4}{*}{MTT} & ConvNet & 21.5 & 28.2 & 34.8 & 40.5 & 0.60/0.64/0.6 & 0.65/0.71/1.8 & 0.73/0.81/9.5 & 0.79/0.86/18.9 & \multirow{20}{*}{81.5} & \multirow{20}{*}{83.2} & \multirow{20}{*}{82.7} & \multirow{20}{*}{84.1} & 0.54/1048.2 & 0.56/1032.7 & 0.55/1082.5 & 0.59/1023.4\\
    & AlexNet & 20.7 & 26.3 & 35.6 & 38.7 & 0.63/0.68/0.9 & 0.68/0.74/1.9 & 0.77/0.85/10.4 & 0.77/0.85/16.0 &  &  &  &  & 0.55/1103.7 & 0.55/1079.2 & 0.56/1043.9 & 0.58/1037.5\\
    & ResNet18 & 22.4 & 30.1 & 33.5 & 41.9 & 0.62/0.67/1.1 & 0.65/0.71/1.3 & 0.73/0.80/7.4 & 0.81/0.88/24.8 &  &  &  &  & 0.53/1072.4 & 0.54/1034.0 & 0.56/1052.1 & 0.56/1089.2\\
    & VGG11-BN & 22.5 & 29.2 & 36.2 & 40.7 & 0.63/0.69/1.4 & 0.66/0.72/1.9 & 0.76/0.80/14.1 & 0.79/0.87/20.9 &  &  &  &  & 0.52/1132.1 & 0.54/1058.6 & 0.55/1075.2 & 0.58/1035.2\\

    \cline{1-10}
    \cline{15-18}
    \multirow{4}{*}{FTD} & ConvNet & 22.4 & 26.5 & 37.6 & 42.8 & 0.63/0.68/0.5 & 0.68/0.75/3.6 & 0.74/0.82/9.6 & 0.77/0.85/16.2 &  &  &  &  & 0.55/1088.9 & 0.56/1112.4 & 0.60/1073.3 & 0.61/1067.3\\
    & AlexNet & 22.1 & 25.7 & 31.8 & 38.5 & 0.62/0.67/0.7 & 0.63/0.68/2.7 & 0.73/0.81/6.3 & 0.75/0.83/15.1 &  &  &  &  & 0.56/1097.3 & 0.56/1047.3 & 0.59/1056.9 & 0.60/1018.3\\
    & ResNet18 & 19.3 & 25.9 & 36.6 & 40.5 & 0.59/0.63/0.3 & 0.69/0.74/4.0 & 0.72/0.80/10.5 & 0.77/0.84/14.5 &  &  &  &  & 0.55/1107.2 & 0.56/1092.1 & 0.59/1067.4 & 0.60/1043.5\\
    & VGG11-BN & 21.4 & 24.3 & 35.7 & 41.6 & 0.61/0.65/0.4 & 0.67/0.73/2.8 & 0.77/0.86/15.1 & 0.76/0.85/15.2 &  &  &  &  & 0.56/1068.4 & 0.56/1062.9 & 0.58/1036.9 & 0.61/1028.9\\

    \cline{1-10}
    \cline{15-18}
    \multirow{4}{*}{DATM} & ConvNet & 26.4 & 32.6 & 44.6 & 51.3 & 0.66/0.72/2.6 & 0.74/0.82/10.2 & 0.76/0.84/12.2 & 0.82/0.90/30.1 &  &  &  &  & 0.53/1089.2 & 0.57/1056.8 & 0.60/992.5 & 0.61/987.3\\
    & AlexNet & 24.5 & 33.7 & 42.8 & 50.5 & 0.68/0.75/3.1 & 0.75/0.83/9.1 & 0.75/0.83/8.7 & 0.80/0.88/16.1 &  &  &  &  & 0.55/1058.2 & 0.56/1043.8 & 0.59/1014.7 & 0.62/954.1\\
    & ResNet18 & 28.3 & 35.1 & 44.3 & 54.2 & 0.70/0.77/4.2 & 0.76/0.85/11.7 & 0.75/0.83/14.1 & 0.81/0.89/25.9 &  &  &  &  & 0.54/1077.5 & 0.57/1048.2 & 0.59/1021.1 & 0.63/930.4\\
    & VGG11-BN & 28.1 & 33.2 & 45.0 & 52.5 & 0.69/0.76/1.4 & 0.75/0.83/8.6 & 0.77/0.85/16.1 & 0.81/0.89/26.5 &  &  &  &  & 0.53/1089.0 & 0.55/1036.4 & 0.61/990.6 & 0.62/966.3\\

    \cline{1-10}
    \cline{15-18}
    \multirow{4}{*}{SelMatch} & ConvNet & 32.2 & 38.6 & 48.1 & 54.8 & 0.74/0.82/9.7 & 0.76/0.83/15.7 & 0.79/0.87/23.3 & 0.82/0.89/28.5 &  &  &  &  & 0.54/1067.4 & 0.55/1036.5 & 0.59/991.3 & 0.62/924.5\\
    & AlexNet & 31.3 & 39.9 & 46.7 & 52.3 & 0.73/0.80/8.9 & 0.77/0.85/12.2 & 0.78/0.86/17.8 & 0.81/0.89/25.4 &  &  &  &  & 0.54/1064.9 & 0.54/1084.6 & 0.58/1003.5 & 0.61/956.3\\
    & ResNet18 & 30.6 & 37.5 & 49.3 & 55.4 & 0.72/0.79/4.3 & 0.78/0.86/17.4 & 0.80/0.88/21.7 & 0.83/0.91/27.5 &  &  &  &  & 0.53/1078.4 & 0.57/1028.5 & 0.60/988.6 & 0.62/927.3\\
    & VGG11-BN & 30.5 & 38.4 & 49.7 & 53.3 & 0.74/0.82/7.1 & 0.79/0.86/17.3 & 0.79/0.87/17.1 & 0.82/0.89/18.1 &  &  &  &  & 0.54/1089.4 & 0.55/1042.6 & 0.58/1014.2 & 0.60/936.0\\

    \cline{1-10}
    \cline{15-18}
    \multirow{4}{*}{SeqMatch} & ConvNet & 27.2 & 32.8 & 39.2 & 43.0 & 0.70/0.78/5.2 & 0.73/0.81/8.8 & 0.77/0.85/16.1 & 0.77/0.86/17.1 &  &  &  &  & 0.54/1089.2 & 0.54/1077.5 & 0.57/1037.6 & 0.58/1024.5\\
    & AlexNet & 27.4 & 31.6 & 37.2 & 42.1 & 0.71/0.77/2.9 & 0.72/0.80/8.9 & 0.75/0.83/12.9 & 0.77/0.85/16.1 &  &  &  &  & 0.52/1143.1 & 0.54/1086.3 & 0.55/1078.3 & 0.57/1032.5\\
    & ResNet18 & 28.5 & 35.2 & 38.1 & 45.9 & 0.69/0.76/2.6 & 0.74/0.81/7.5 & 0.74/0.81/12.7 & 0.78/0.86/20.9 &  &  &  &  & 0.52/1092.3 & 0.55/1053.1 & 0.56/1044.8 & 0.58/1035.4\\
    & VGG11-BN & 29.1 & 33.9 & 40.4 & 44.5 & 0.71/0.77/3.7 & 0.74/0.82/7.2 & 0.77/0.85/13.6 & 0.78/0.86/17.4 &  &  &  &  & 0.52/1102.3 & 0.53/1088.4 & 0.56/1033.3 & 0.58/1051.2\\
    \bottomrule
\end{tabular}\label{tab:eval_cifar100}
\end{table*}

\begin{table*}[h!]
\caption{Evaluation of the IRA on the SOTA Dataset Distillation Algorithms on the TinyImageNet-200 Dataset.}
\centering
\footnotesize
\renewcommand{\arraystretch}{1.2}
\setlength{\tabcolsep}{2pt}
\begin{tabular}{cc|ccc|ccc|ccc|ccc}
    \toprule
    \multirow{4}{*}{Method} & \multirow{4}{*}{Model Arch.} & \multicolumn{3}{c|}{Test Top-1 Acc \%} & \multicolumn{3}{c|}{MIA BA$\uparrow$/AUC$\uparrow$/T@LF \%$\uparrow$} & \multicolumn{3}{c|}{AIA Top-1 Acc \%$\uparrow$} & \multicolumn{3}{c}{MIV Atk. Acc.$\uparrow$/KNN Dist.$\downarrow$}\\
    & & \multicolumn{3}{c|}{IPC} & \multicolumn{3}{c|}{IPC} & \multicolumn{3}{c|}{IPC} & \multicolumn{3}{c}{IPC}\\
    & & 10 & 25 & 50 & 10 & 25 & 50 & 10 & 25 & 50 & 10 & 25 & 50\\

    \midrule

    \multirow{4}{*}{MTT} & ConvNet & 15.3 & 19.7 & 21.5 & 0.65/0.70/0.7 & 0.68/0.75/1.8 & 0.71/0.78/3.3 & \multirow{20}{*}{  75.2  } & \multirow{20}{*}{  73.6  } & \multirow{20}{*}{ 
 77.0  } & 0.32/1359.2 & 0.34/1344.2 & 0.36/1306.8 \\
    & AlexNet & 15.1 & 17.8 & 19.9 & 0.63/0.68/1.3 & 0.66/0.71/1.6 & 0.70/0.77/2.8 &  &  &  & 0.33/1382.1 & 0.35/1363.7 & 0.35/1327.8 \\
    & ResNet18 & 14.7 & 19.4 & 22.1 & 0.62/0.69/0.9 & 0.68/0.74/2.2 & 0.71/0.78/3.5 &  &  &  & 0.32/1402.8 & 0.34/1352.8 & 0.35/1323.5 \\
    & VGG11-BN & 15.4 & 18.6 & 21.7 & 0.61/0.66/1.0 & 0.67/0.74/2.8 & 0.69/0.75/3.7 &  &  &  & 0.32/1376.9 & 0.34/1368.9 & 0.36/1324.9 \\

    \cline{1-8}
    \cline{12-14}
    \multirow{4}{*}{FTD} & ConvNet & 18.9 & 23.4 & 26.9 & 0.65/0.71/3.6 & 0.72/0.80/8.3 & 0.74/0.82/7.3 &  &  &  & 0.33/1374.6 & 0.34/1324.7 & 0.36/1309.4\\
    & AlexNet & 18.1 & 21.8 & 23.5 & 0.69/0.76/2.5 & 0.71/0.77/5.6 & 0.72/0.79/5.7 &  &  &  & 0.35/1342.8 & 0.35/1347.3 & 0.36/1318.4\\
    & ResNet18 & 19.4 & 24.3 & 27.7 & 0.70/0.77/4.5 & 0.74/0.82/9.1 & 0.74/0.81/9.2 &  &  &  & 0.35/1328.4 & 0.36/1320.7 & 0.36/1302.8\\
    & VGG11-BN & 18.5 & 23.7 & 26.4 & 0.68/0.74/3.2 & 0.73/0.81/6.7 & 0.74/0.82/7.3 &  &  &  & 0.34/1342.7 & 0.35/1311.7 & 0.36/1321.4\\

    \cline{1-8}
    \cline{12-14}
    \multirow{4}{*}{DATM} & ConvNet & 23.5 & 32.6 & 38.3 & 0.72/0.79/4.1 & 0.77/0.85/11.4 & 0.79/0.87/15.9 &  &  &  & 0.35/1325.6 & 0.36/1287.4 & 0.37/1253.8\\
    & AlexNet & 24.8 & 31.0 & 36.7 & 0.73/0.80/7.0 & 0.77/0.85/11.1 & 0.78/0.86/13.2 &  &  &  & 0.34/1347.8 & 0.36/1317.2 & 0.36/1299.4\\
    & ResNet18 & 23.0 & 33.6 & 38.5 & 0.72/0.79/5.2 & 0.78/0.85/13.2 & 0.78/0.86/14.8 &  &  &  & 0.35/1331.5 & 0.36/1307.8 & 0.36/1297.0\\
    & VGG11-BN & 24.2 & 31.8 & 37.2 & 0.71/0.78/3.8 & 0.76/0.84/9.3 & 0.79/0.87/12.8 &  &  &  & 0.35/1326.9 & 0.35/1343.9 & 0.37/1285.4\\

    \cline{1-8}
    \cline{12-14}
    \multirow{4}{*}{SelMatch} & ConvNet & 25.9 & 34.2 & 40.5 & 0.75/0.83/10.5 & 0.77/0.85/12.8 & 0.79/0.87/18.3 &  &  &  & 0.35/1337.1 & 0.35/1346.2 & 0.37/1264.8\\
    & AlexNet & 25.7 & 33.7 & 38.9 & 0.74/0.82/6.7 & 0.76/0.85/13.1 & 0.78/0.86/17.0 &  &  &  & 0.34/1351.2 & 0.35/1322.2 & 0.35/1298.4\\
    & ResNet18 & 26.3 & 34.9 & 41.8 & 0.75/0.83/9.1 & 0.78/0.86/12.3 & 0.79/0.87/16.1 &  &  &  & 0.33/1368.5 & 0.35/1317.3 & 0.36/1289.5\\
    & VGG11-BN & 24.2 & 35.4 & 39.9 & 0.74/0.82/9.3 & 0.77/0.85/15.4 & 0.79/0.87/19.7 &  &  &  & 0.34/1363.2 & 0.36/1306.7 & 0.36/1294.8\\

    \cline{1-8}
    \cline{12-14}
    \multirow{4}{*}{SeqMatch} & ConvNet & 23.0 & 28.6 & 34.9 & 0.71/0.79/5.2 & 0.74/0.82/9.3 & 0.77/0.84/11.5 &  &  &  & 0.33/1365.9 & 0.34/1354.6 & 0.35/1327.1\\
    & AlexNet & 23.1 & 30.1 & 33.7 & 0.73/0.80/7.6 & 0.75/0.83/9.6 & 0.78/0.86/13.7 &  &  &  & 0.33/1358.3 & 0.34/1363.8 & 0.35/1348.3\\
    & ResNet18 & 21.8 & 29.7 & 33.8 & 0.71/0.77/3.8 & 0.76/0.83/7.4 & 0.77/0.86/9.7 &  &  &  & 0.34/1382.0 & 0.35/1371.8 & 0.35/1362.1\\
    & VGG11-BN & 22.5 & 30.7 & 34.3 & 0.72/0.79/7.7 & 0.76/0.84/10.6 & 0.76/0.84/10.3 &  &  &  & 0.34/1361.7 & 0.34/1350.1 & 0.36/1305.8\\
    \bottomrule
\end{tabular}\label{tab:eval_tinyimagenet}
\end{table*}

\begin{figure*}[h!]
    \centering
    \includegraphics[width=\linewidth]{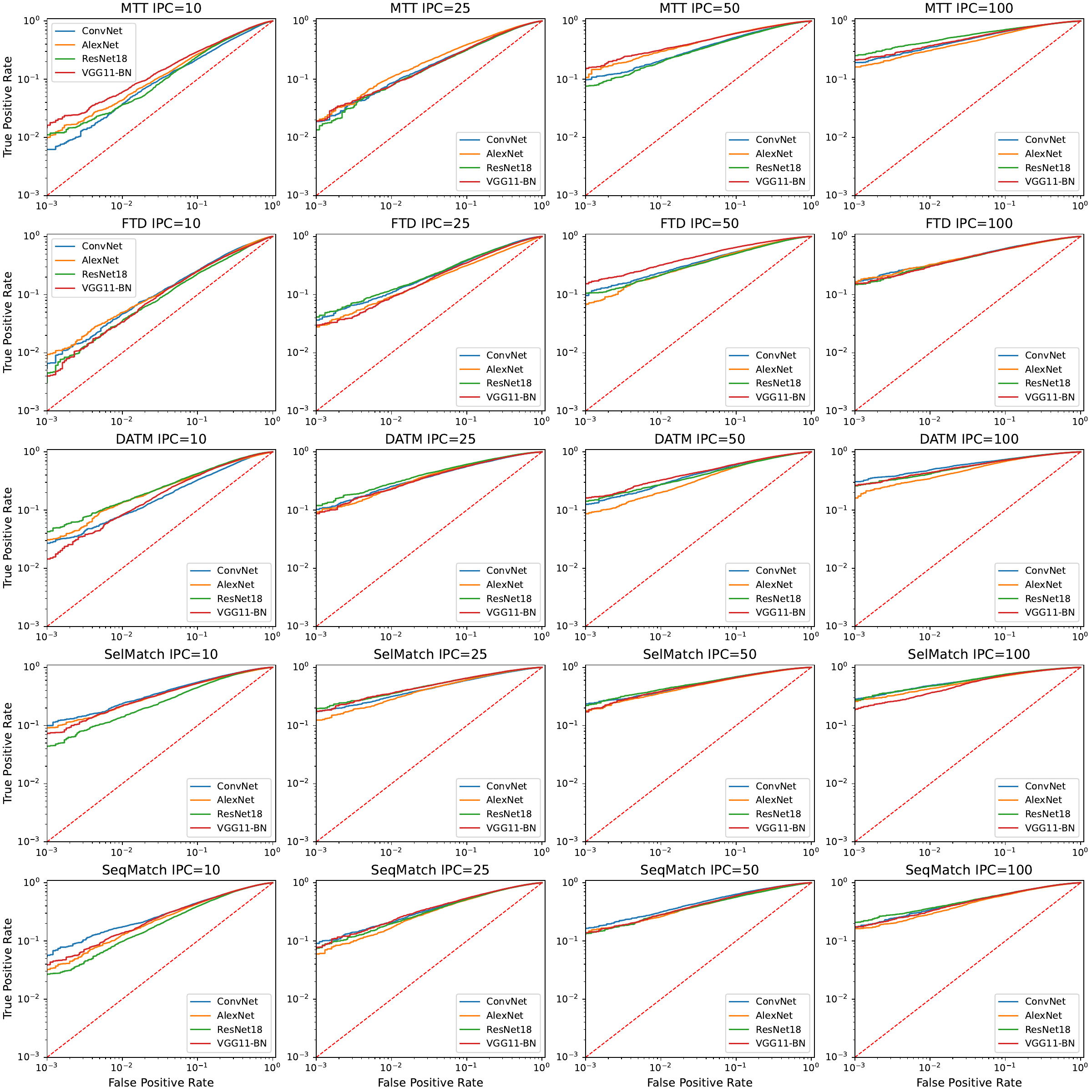}
    \caption{Log-scale ROC curves for the MIAs on the synthetic CIFAR-100 dataset. \textit{Please zoom in to see the details due to the page limit.}}
    \label{fig:atk_eval_cifar100}
\end{figure*}

\begin{figure*}[h!]
    \centering
    \includegraphics[width=.9\linewidth]{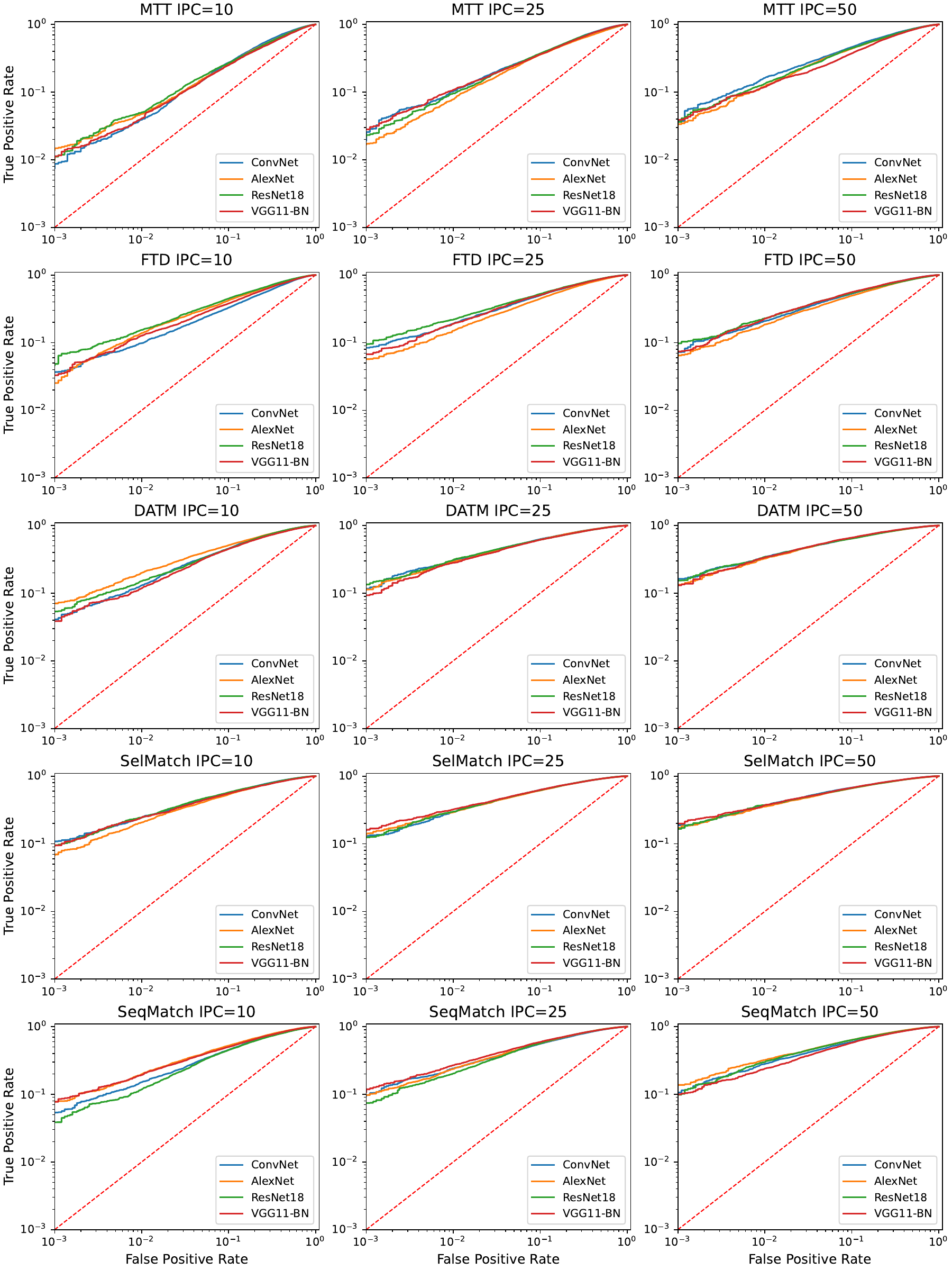}
    \caption{Log-scale ROC curves for the MIAs on the synthetic TinyImagenet-200 dataset. \textit{Please zoom in due to the page limit.}}
    \label{fig:atk_eval_tinyimagenet}
\end{figure*}


\end{document}